\documentclass[alpha-refs]{wiley-article}

\usepackage{siunitx}
\usepackage{bm}
\usepackage{enumitem}
\usepackage{hyperref}
\usepackage{cleveref}
\usepackage{longtable}
\usepackage{amsmath}

\papertype{Original Article}
\paperfield{Journal Section}
\definecolor{light-gray}{HTML}{E5E4E2}
\definecolor{light-gray2}{gray}{0.85}

\title{ GPMatch: A Bayesian Doubly Robust Approach to Causal Inference with Gaussian Process Covariance Function As a Matching Tool
}

\author[1\authfn{1}, 2\authfn{2}]{Bin Huang PhD}
\author[1\authfn{1}]{Chen Chen PhD}
\author[3\authfn{3}]{Jinzhong Liu PhD}

\affil[1]{Division of Biostatistics and Epidemiology, Cincinnati Children's Hospital Medical Center, Cincinnati, OH, 45229, USA}
\affil[2]{Department of Pediatrics, University of Cincinnati College of Medicine, Cincinnati, OH, 45267, USA}
\affil[3]{MedPace Inc, Cincinnati, OH, 45227, USA}

\corraddress{Bin Huang PhD}
\corremail{bin.huang@cchmc.org}

\fundinginfo{Patient Centered Outcome Research Institute, PCORI ME-1408-19894; The Center for Clinical and Translational Science and Training, National Center for Advancing Translational Sciences of the National Institutes of Health, Award Number 5UL1TR001425-03 }

\runningauthor{Huang, B. et al.}

\begin{document}

\maketitle

\begin{abstract}
Gaussian process (GP) covariance function is proposed as a matching tool in GPMatch  within a full Bayesian framework under relatively weaker causal assumptions. The matching is accomplished by utilizing GP prior covariance function to define matching distance. We show that GPMatch provides a doubly robust estimate of the averaged treatment effect (ATE) much like the G-estimation, the ATE is correctly estimated when either conditions are satisfied: 1) the GP mean function correctly specifies potential outcome \(Y^{(0)}\); or 2) the GP covariance function correctly specifies matching structure.  Simulation studies were carried out without assuming any known matching structure nor functional form of the outcomes. The results demonstrate that GPMatch enjoys well calibrated frequentist properties, and outperforms many widely used methods including Bayesian Additive Regression Trees. The case study compares effectiveness of early aggressive use of biological medication in treating children with newly diagnosed Juvenile Idiopathic Arthritis, using data extracted from electronic medical records. 

\end{abstract}

\section{ Introduction }
Data from nonrandomized experiments, such as registry and electronic records, are becoming indispensable sources for answering causal inference questions in health, social, political, economics and many other disciplines. Under the assumptions of ignorable treatment assignment and distinct model parameters governing the science and treatment assignment mechanisms, \cite{Rubin1978} showed the Bayesian inference of causal treatment effect can be approached by directly outcome modeling, treating it as a missing potential outcome problem. Direct modeling is able to utilize the many Bayesian regression modeling techniques to address complex data type and data structures, such as examples in \cite{Hirano2000}, \cite{Zajonc2012}, \cite{Imbens1997} and \cite{Baccini2017}. 

Parameter rich Bayesian modeling techniques are particularly appealing as it does not presume an known functional form, thus may help mitigate potential model miss-specification issues. \cite{Hill2011} suggested Bayesian additive regression tree (BART) can be used for causal inference, and showed it produced more accurate estimates of average treatment effects compared to propensity score matching, inverse propensity weighted estimators, and regression adjustment in the nonlinear setting, and performed as well under the linear setting.  Others have used Gaussian Process in conjunction with Dirichlet Process priors, e.g. \cite{roy2016bayesian} and \cite{Xu2016}.  \cite{roy2017bayesian} devised enriched Dirichlet Process priors tackling missing covariate issues. However, naive use of regression techniques could lead to substantial bias in estimating causal effect as demostrated in \cite{Hahn2018}.  

The search for ways of incorporating propensity of treatment selection into the Bayesian causal inference has been long standing. Including propensity score (PS) as a covariate into the outcome model may be a natural way. However, joint modeling of outcome and treatment selection models leads to a “feedback” issue, and a two-stage approach was suggested by \cite{mccandless2010cutting},  \cite{zigler2013model} and many others. Discussion about whether the uncertainty of the first step propensity score modeling should be taken into account when obtaining the final result in the second step can be found in \cite{hill2006interval}, \cite{Ho2007}, \cite{Rubin2006}, \cite{Rubin1996} for details. \cite{Saarela2016} proposed an approximate Bayesian approach incorporating inverse probability treatment assignment probabilities as importance-sampling weights in Monte Carlo integration. It offers a Bayesian version to the augmented inverse probability treatment weighting (AIPTW). \cite{Hahn2017} suggested incorporating estimated treatment propensity into the regression to explicitly induce covariate dependent prior in regression model. These methods all require a separate step of treatment propensity modeling, thus may suffer if the propensity model is mis-specified. 

Matching is one of the most sought-after method used for designing observational study to answer causal questions. Matching experimental units on their pre-treatment assignment characteristics helps to remove the bias by ensuring the similarity or balance between the experimental units of the two treatment groups. Matching methods impute the missing potential outcome with the value from the nearest match or the weighted average of the values within the nearby neighborhood defined by (a chosen value) caliper. Matching on multiple covariates could be challenging when the dimension of the covariates are large. For this reason, matching is often performed using the estimated propensity score (PS) or by the Manhalanobis distance (MD). The idea is, under the no unmeasured confounder setting, matching induces balance between the treated and untreated groups. Therefore, it serves to transform a nonrandomized study into a pseudo randomized study. There are many different matching techniques, a comprehensive review is  provided in \cite{Stuart2010}. A recent study by \cite{226731} compared the PS matching with the MD matching and suggests that PS matching can result more biased and less accurate estimate of averaged causal treatment as the precision of matching improves, while the MD matching is showing improved accuracy.  Common to matching methods, the data points without a match are discarded. Such a practice may lead to a sample no longer representative of the target population. A user-specified caliper is often required, but different calipers could lead to very different results. Furthermore, matching on a miss-specified PS could lead to invalid causal inference results. 

\cite{rubin1973use} suggested that the combination of matching and regression is a better approach than using either of them alone. \cite{Ho2007} advocated matching as nonparametric preprocessing for reducing dependence on parametric modeling assumptions. \cite{Gutman2017} examined different strategies of combining the preprocessed matching with a regression modeling of the outcome through extensive simulation studies. They demonstrated that some commonly used causal inference methods have poor operating characteristics, and suggested regression modeling after pre-processed matching works better. To our knowledge, no existing method can accomplish matching and regression modeling in a single step. 

Gaussian process (GP) prior has been widely used to describe biological, social, financial and physical phenomena, due to its ability to model highly complex dynamic system and its many desirable mathematical properties. Recent literature, e.g.\cite{Choi2013} and \cite{Choi2007}, has established posterior consistency for Bayesian partially linear GP regression models. Bayesian modeling with GP prior can be viewed as a marginal structural model where the potential outcome under the non-intervention condition \(Y^{(0)}\) is modeled non-parametrically. It allows for predicting the missing response by a weighted sum of observed data, with larger weights assigned to those in closer proximity but smaller to those further away, much like a matching procedure. This motivated us to consider using GP prior covariance function as a matching tool for Bayesian causal inference.  

The idea of utilizing GP prior in Bayesian approach to causal inference is not new.  Examples can be found in \cite{roy2016bayesian} for addressing heterogeneous treatment effect, in \cite{Xu2016} for handling dynamic treatment assignment, and in \cite{roy2017bayesian} for tackling missing data. While these studies demonstrated GP prior could be used to achieve flexible modeling and tackle complex setting, no one has considered GP as a matching tool. This study adds to the literature in several ways. First, we offer a principled approach to Bayesian causal inference utilizing GP prior covariance function as a matching tool, which accomplishes matching and flexible outcome modeling in a single step. Second, we provide relaxed causal assumptions than the widely adopted assumptions from the landmark paper by \cite{rosenbaum1983}. By admitting additional random errors in outcomes and in the treatment assignment, these new assumptions fit more naturally within Bayesian framework. Under these weaker causal assumptions, GPMatch method offers a doubly robust approach in the sense that the averaged causal treatment effect is correctly estimated when either one of the conditions are met : 1) when the mean function correctly specifies the \(Y^{(0)}\); or 2) the covariance function matrix correctly specifies the matching structure. At last, the proposed method has been implemented in an easy-to-use publicly available on-line application (\url{https://pcats.research.cchmc.org/}).  

The rest of the presentation is organized as follows. Section 2 describes methods, where we present problem setup, causal assumptions, and the model specifications. The utility of GP covariance function as a matching tool is presented in Section 3, followed by discussions of its doubly robustness property. Simulation studies are presented in Section 4. Simulations are designed to represent the real world setting where the true functional form is unknown, including the well-known simulation design suggested by \cite{Kang2007}. We compared the GPMatch approach with some commonly used causal inference methods, i.e. linear regression with PS adjustment, AIPTW, and BART, without assuming any knowledge of the true data generating models. The results demonstrate that the GPMatch enjoys well calibrated frequentist properties, and outperforms many widely used methods under the dual miss-specification setting. Section 5 presents a case study, examining the comparative effectiveness of an early aggressive use of biological medication in treating children with recent diagnosed juvenile idiopathic arthritis (JIA). Section 6 presents summary, discussions and future directions.

\section{Method}
\subsection{Problem Setup and Notations}
\begin{figure}[bt]
\centering
\includegraphics[width=10cm]{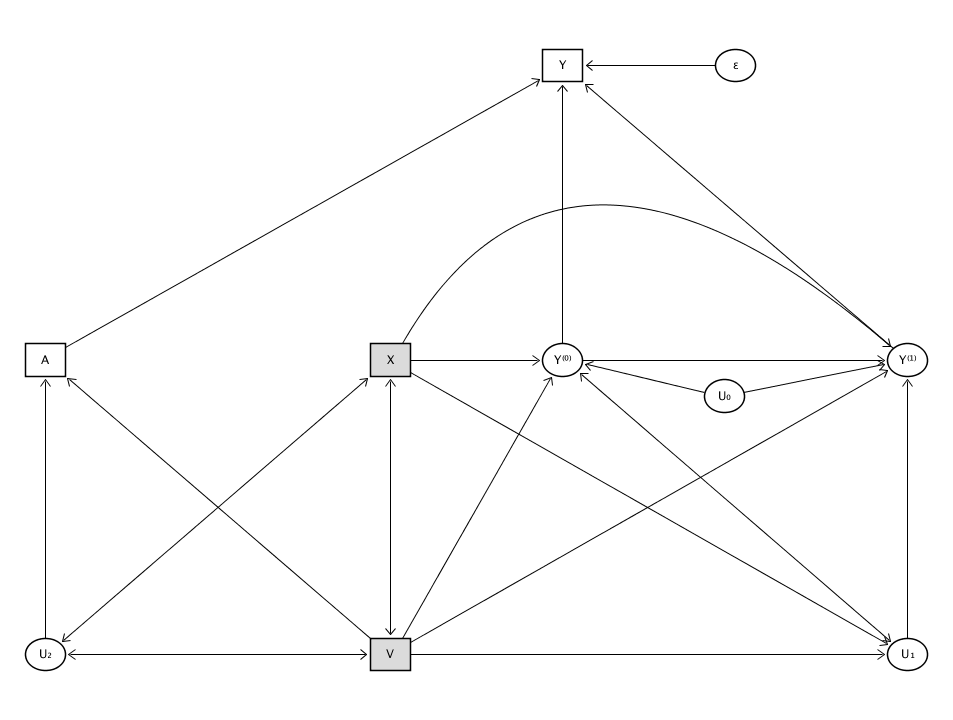}
\caption{The Directed Acyclic Graphic (DAG) Presentation of the Problem Setup}
\label{fig:0}
\end{figure}
The problem setup is depicted in the directed acyclic graphic (DAG), where the rectangular nodes are measured and oval nodes are latent or unmeasured variables. The \(\boldsymbol{X}\) and \(\boldsymbol{V}\) 
are observed covariates, and \(\boldsymbol{Y}\) is the observed outcome. The treatment assignment (\(A = 0/1\)) is a binary indicator, where 0 indicating comparator or the nature occurring condition and  1 indicates intervention. Correspondingly, the potential outcomes   \((Y^{(0)}, Y^{(1)}) \), are two latent variables. The unmeasured covariates are denoted by \(U_0, U_1, U_2\),  representing three types of unmeasured covariates for \(Y^{(0)}, Y^{(1)} \) and \(A\) correspondingly.  The potential outcome \( Y^{(0)} \) under the controlled condition is determined jointly by \(\boldsymbol{X}\) a p-dimensional and  \( \boldsymbol{V}\) a q-dimensional vector of the observed covariates plus a unmeasured covariate  \(U_0\). Thus, \((\boldsymbol{X}, \boldsymbol{V}, U_0)\) are  prognostic variables. Similarly, the potential outcome \( Y^{(1)} \) under the intervention condition is determined jointly by the observed covariates \(( \boldsymbol{X}, \boldsymbol{V} )\) 
and the unobserved covariates (\(U_0, U_1)\). 
The observed outcome \(Y\) is a noisy version of the corresponding potential outcomes, with an error term \(\epsilon\).  The treatment is assigned according to an unknown propensity score, which is determined by the baseline covariates  \( \boldsymbol{V}\) 
and  \(U_2\).  The observed baseline covariates \(\boldsymbol{X}\)and \(\boldsymbol{V} \) 
could be overlapping, whereas different symbols are used to distinguish their roles in science mechanisms and the treatment assignment process respectively. For example, X could include patient age, gender, genetic makeup, family disease history, past and current medication use as well laboratory results and other disease characteristics, which are directly related to the prognosis of the disease. The V could include the above X variable, as well as other considerations to the treatment decisions including insurance, social economic status of patient family, education and clinical centers. Most of these important X and V covariates are available in a patient registry and electronic medical records, thus are observable. Other factors could play a role in treatment decisions, such as patient and clinician's personal preferences, culture believes and past experiences. However, they are almost never recorded. These factors are collective referred as \(U_2\). The residual terms of responses \((\epsilon. U_0, U_1)\) can be overlapping or correlated, the corresponding links are omitted in the Figure 1 for better visual presentation. 

The DAG can be expressed by a set of structural equation models:
\begin{equation}
\begin{array}{l}
Y_i = A_i Y_i^{(1)}+(1-A_i) Y_i^{(0)} + \epsilon_i,   \label{eq:a1}\\
Y_i^{(a)}=f^{(0)}(\boldsymbol{x_i},\boldsymbol{v_i})+ a (\tau (\boldsymbol{x_i}) + u_{1i}) + u_{0i}\\
Pr(A_i) = \pi(v_i, u_{2i}) \\
\end{array}
\end{equation}
where  \(E(\epsilon_i) =0\) and \( E(u_{ki})=0, \) for \(k=0,1,2.\) To ensure the causal treatment effect can be estimated without bias, the following condition must be true:  \(\epsilon \bot (Y^{(0)},Y^{(1)})\), \( (U_0, U_1) \bot A|\boldsymbol{X,V}\),  \( U_2 \bot \epsilon\) and  \( U_2 \bot  Y | A, \boldsymbol{X},\boldsymbol{V}. \) Violation of any of these conditional independence condition can open up the back-door path from \(Y\) to \(A\) (\cite{pearl2009causality}). The \(f^{(0)}(\cdot)\), \(f^{(1)}(\cdot)\) and \(\pi(\cdot)\) are unknown functions that describes the potential outcome science mechanism and treatment assignment process.   The sample averaged treatment effect of  all individual level effect 
\( \tau_i = \tau(x_i) + u_{1i}\),
\(\tau=\frac{1}{n} \sum_i \tau_i\)  is the parameter of interest, which is referred as the averaged treatment effect (ATE).  


\subsection{The Causal Assumptions}
The causal assumptions are necessary to ensure unbiased estimate of casual treatment effect. The causal assumptions are presented in the DAG and the structural equation (1). 
Notice that the DAG includes three types of unmeasured covariates, where \(U_0\) indicates unknown correlation between the pair of potential outcomes,  \(U_1\) a potential lurking variable, and  \(U_2\)  a potential confounding variable.   Under the conditional independence conditions, the observed covariates \((\boldsymbol{X}, \boldsymbol{V})\) is a minimum sufficient set  for identifying causal treatment effect. Further, with assuming distinct model parameters, it is relatively straight forward to see that the posterior of the potential outcomes can be derived directly by 
\begin{equation*}
[Y^{(0)},Y^{(1)} | A, \boldsymbol{X,V}, Y] = \frac{[Y, Y^{(0)}, Y^{(1)} | A, \boldsymbol{X,V}]}{[Y|A,\boldsymbol{X,V}]}.
\end{equation*}Comparing to  the widely adopted causal assumptions laid out in the landmark paper by the \cite{rosenbaum1983} (RR), the DAG presents a weaker version of causal assumptions :
\begin{enumerate}[label=CA\arabic*.]
\item \label{CA1} Instead of the stable unit treatment value assumption (SUTVA), we assume stable unit treatment value expectation assumption (SUTVEA).  Specifically,
\begin{enumerate}[label=(\roman*)]
\item The consistency assumption of RR requires the observed outcome is an exact copy of the potential outcome, i.e. \(Y_i=Y_i^{(0)}(1-A_i)+Y_i^{(1)} A_i\). Instead, only  \(Y \bot \boldsymbol{X,V} | (A, Y^{(0)}, Y^{(1)})\) is required. In other words, we consider the observed outcome is a noisy copy of the potential outcome where the expectation of the observed outcome \(E(Y_i)=Y_i^{(0)}(1-A_i)+Y_i^{(1)} A_i.\)    
\item The no interference assumption of RR requires the potential outcomes of one experiment unit is not influenced by the  potential outcomes of another experiment unit, i.e. \(Y_i^{(a)} \bot Y_j^{(b)}\).  Instead, we assume the observed outcomes from different units are conditional independence given the observed covariates \(Y_i \bot Y_j | A, \boldsymbol{X},\boldsymbol{V}\)
\end{enumerate}
The SUTVEA assumption acknowledges existence of residual random error in the outcome measure.  The observed outcomes may differ from the corresponding true potential outcomes due to some measurement errors. In addition, the observed outcomes could differ when treatment received deviates from its intended version of treatment. For example, outcomes could differ by the timing of the treatment, pre-surgery preparation procedure or the concomitant medication. In addition, we consider the potential outcomes from different experimental units may be correlated, where the correlations are determined by the covariates. Since only one outcome could be observed out of all potential outcomes, the causal inference presents a highly structured missing data setup where the correlations between \( (Y_i^{(1)}, Y_i^{(0)})\) are not directly identifiable.  Admitting residual errors and allowing for explicit modeling of the covariance structure, the new assumptions could facilitate better statistics inference.
\vspace{5mm}
\item \label{CA2}  Similarly as in RR, we assume ignorable treatment assignment assumption \( [Y^{(a)} | A=1, \boldsymbol{X,V}] =[Y^{(a)} | A=0, \boldsymbol{X,V}], \) for \(a = 0,1\). That is the marginal distribution of a potential outcome can be obtained by modeling the observed covariates only, indepedent from the treatment assignment. 
As depicted in DAG, presence of  unmeasured confounder is admissible, as long as the back-door path from \(Y\) to \(A\) is blocked by the observed covariates. In practice, it is almost never possible for us to capture all the considerations factored into a treatment decisions, such as personal preferences and past experiences. However, it is reasonable to consider the uncounted residual error in treatment assignment, conditional on the observed covariates (e.g. patient demographics, insurance, disease characteristics,  laboratory and medical diagnostic tests), is not related to the potential outcomes.   
\vspace{5mm}
\item Positivity Assumption.  Same as in RR, we assume every sample unit has nonzero probability of being assigned into either one of the treatment arms, i.e. \(0<Pr(A_i|\boldsymbol{x_i,v_i})<1\) for \( \forall (\boldsymbol{x_i, v_i} ) \).  This assumption is adopted to ensure the equipoise of the causal inference. 

\end{enumerate}

  
\subsection{The GPMatch Model Specifications}

Marginal structural model (MSM) is a widely adopted modeling approach to causal inference, which serves as a natural framework for Bayesian causal inference. The MSM specifies 
\begin{equation*}
Y_i^{(1)} = Y_i^{(0)}+A_i\tau_i.
\end{equation*}
Without prior knowledge about the true functional form, we let \( Y_i^{(0)} \sim GP(\boldsymbol{\mu_f}, \boldsymbol{K})  \), where the mean function \(\mu_f\) maybe modeled by a parametric regression equation, and  \(\boldsymbol{K} \) defines the covariance function of the GP prior. Specifically, GPMatch is proposed as a partially linear Gaussian process regression fitting to the observed outcomes, 
\begin{equation} \label{eq:d1}
Y_i=f_i(\boldsymbol{x_i}, \boldsymbol{v_i})+A_i\tau(\boldsymbol{x_i})+\epsilon_i,
\end{equation}
where 
\begin{equation*}
\begin{array} {lr}
f_i(\boldsymbol{x_i}, \boldsymbol{v_i}) = \mu_f(\boldsymbol{x_i}) + \eta(\boldsymbol{v_i}),\\
\eta_i(\boldsymbol{v_i}) \sim GP(0, \boldsymbol{K}),\\
\epsilon_i \sim N(0, \sigma_0),\\
\epsilon_i \perp \eta_i.
\end{array}
\end{equation*}
Here, we may let \(\boldsymbol{\mu_f}=((1,\boldsymbol{{X_i}'})\beta)_{n\times 1}\), where \(\boldsymbol{\beta} \) is a \( (1+p) \) dimension parameter vector of regression coefficients for the mean function. This is to allow for implementing any existing knowledge about the prognostic determinants to the outcome. Also, let \(\boldsymbol{\tau} = \left( (1 , \boldsymbol{X_i}') \boldsymbol{\alpha} \right)_{n\times 1} \) to allow for potential heterogeneous treatment effect, where \(\boldsymbol{\alpha} \) is a \( (1+p) \) dimension parameter vector of regression coefficients for the treatment effect. 

Let \(\boldsymbol{Y_{n}}=(Y_{i})_{n\times 1}\), the model \eqref{eq:d1} can be re-expressed in a multivariate representation 
\begin{equation} 
\boldsymbol{Y_{n}}|\boldsymbol{A,X,V,\gamma} \sim MVN(\boldsymbol{Z'\gamma,\Sigma }),\label{eq:d2}
\end{equation}
where \(\boldsymbol{Z}'=(1, \boldsymbol{X_i}', A_i, A_i\times \boldsymbol{X_i}')_{n\times (2+2 p)}\), \(\boldsymbol{\gamma}=(\boldsymbol{\beta, \alpha})\), \(\boldsymbol{\Sigma}=(\sigma_{ij})_{n\times n}\), with \(\sigma_{ij}=K(\boldsymbol{v_i},\boldsymbol{v_j})+\sigma_0^2\delta_{ij}\). The \(\delta_{ij}\) is the Kronecker function, \(\delta_{ij}=1\) if \(i=j\), and 0 otherwise.   

Gaussian process can be considered as distribution over function. The covariance function \(\boldsymbol{K} \), where \( k_{ij} = Cov(\boldsymbol{\eta_i},\boldsymbol{\eta_j})\), plays a critical role in GP regression. It can be used to reflect the prior belief about the functional form, determining its shape and degree of smoothness. In the next section, we show for the data comes from an experimental design where the matching structure is known, GP covariance could be formulated to reflect the matching structure. Often, the exact matching structure is not available, a natural choice for the GP prior covariance function \( \boldsymbol{K} \) is the squared-exponential (SE) function, where 
\begin{equation}
K(v_i,v_j)=\sigma_f^2 exp \left( -\sum_{k=1}^q \frac{|v_{ki}-v_{kj}|^2}{\phi_k}\right), \label{eq:d4}
\end{equation}
for \(i,j=1,...,n\). The \((\phi_1,\phi_2,...,\phi_q)\) are the length scale parameters for each of the covariates \(\boldsymbol{V}\). 

There are several considerations in choosing the SE covariance function. The GP regression with SE covariance can be considered as a Bayesian linear regression model with infinite basis functions, which is able to fit a smoothed response surface. Because of the GP's ability to choose the length-scale and covariance parameters using the training data, unlike other flexible models such as splines or the supporting vector machine (SVM), GP regression does not require cross-validation(\cite{Rasmussen2006}). Moreover, SE covariance function provides a distance metric that is similar to Mahalanobis distance, thus it could be served as a matching tool . 

The model specification is completed by specification of the rest of priors. 
\begin{equation*} 
\begin {array}{lr}
    \boldsymbol{\gamma} \sim MVN \left( \boldsymbol{0}, \omega \sigma_{lm}^2 ( \boldsymbol{(Z Z')} )^{-1} \right),\\
    \sigma_0^2 \sim IG(a_0, b_0 ),\\
    \sigma_f^2 \sim IG(a_f, b_f ),\\
    \phi_k \sim IG(a_\phi, b_\phi ).
  \end{array}
\end{equation*}
We set \( \omega = 10^6, a_\phi = b_\phi = 1, a_0 = a_f = 2, b_0 = b_f = \sigma_{lm}^2/2, \sigma_{lm}^2 \) is the estimated variance from a simple linear regression model of \( Y \) on \(A\) and \(X\) for computational efficiency. 
 
The posterior of the parameters can be obtained by implementing a Gibbs sampling algorithm: first sample the covariate function parameters from its posterior distribution\( [ \boldsymbol{\Sigma} | Data,\boldsymbol{\alpha,\beta} ] \); then sample the regression coefficient parameter associated with the mean function from its conditional posterior distribution \( [ \boldsymbol{\alpha,\beta}| Data,\boldsymbol{\Sigma}] \), which is a multivariate normal distribution. The individual level treatment effect can be estimated by \(\hat{\tau}(\boldmath{x_i}) = (1 , \boldmath{X_i})' \hat{\boldsymbol{\alpha}}\) and the averaged treatment effect is estimated by
\(\hat{ATE} = \sum_{i=1}^n \frac{\hat{\tau}(\boldmath{x_i})}{n} \).

\section{Estimate ATE: Connections with Matching and G-estimation} 
\subsection{Design the GP Covariance Function as a Matching Tool}
To demonstrate the utility of the GP covariance function as a matching tool, let us first consider design a covariance function for the known matching data structure . In other words, we assume for any given sample unit, we know who are the matching units.  For simplicity, we consider fitting the data with a simple nonparametric version of the GPMatch,  
\begin{equation}
\boldsymbol{Y_n} \sim MVN(\mu \boldsymbol{1}_n+\tau \boldsymbol{A}_n, \boldsymbol{\Sigma}),\label{eq:d5}
\end{equation}
where \(\boldsymbol{\Sigma = K} + {\sigma_0}^2 \boldsymbol{I_n}\).

With known matching structure, the GP covariance function may present the matching structure by letting \(\boldsymbol{K}=(k_{ij})_{n\times n} \), where \( k_{ij}=1 \) indicates that the pair is completely matched, and \(k_{ij}=0\) if unmatched. A common setting of the matched data can be divided into several blocks of subsample within which the matched data points are grouped together. Subsequently,  we may rewrite the covariance function of the nonparametric GP model \eqref{eq:d5}
as a block diagonal matrix where the \(l^{th}\) block matrix takes the form
\begin{equation*}
\boldsymbol{\Sigma_l} = \sigma^2 \left[(1-\rho)\boldsymbol{I}_{n_l}+\rho \boldsymbol{J}_{n_l}\right],
\end{equation*}
where \(\sigma^2 =1+\sigma_0^2\),  \(\rho = 1/\sigma^2\)and \(\boldsymbol{J}_{n_l}\) denotes the matrix of ones. The parameter estimates of the regression parameters can be derived by
\begin{equation*}
\left(\begin{array}{c}
  \hat{\mu} \\
  \hat{\tau}
 \end{array}\right) = \left[\left(\begin{array}{c}
  \boldsymbol{1}_n' \\
  \boldsymbol{A}_n'
 \end{array}\right) \boldsymbol{\Sigma}^{-1} \left(\begin{array}{lr}
  \boldsymbol{1}_n & \boldsymbol{A}_n
 \end{array}\right) \right]^{-1} \left(\begin{array}{c}
  \boldsymbol{1}_n' \\
  \boldsymbol{A}_n'
 \end{array}\right) \boldsymbol{\Sigma}^{-1} \boldsymbol{Y_n}.
\end{equation*}
It follows that the estimated average treatment effect is,
\begin{equation*}
\hat{\tau}=\frac{\boldsymbol{1}_n' \boldsymbol{\Sigma}^{-1} \boldsymbol{1}_n \boldsymbol{A}_n' \boldsymbol{\Sigma}^{-1} \boldsymbol{Y_n}-\boldsymbol{A}_n' \boldsymbol{\Sigma}^{-1} \boldsymbol{1}_n \boldsymbol{1}_n' \boldsymbol{\Sigma}^{-1} \boldsymbol{Y_n} }{\boldsymbol{1}_n' \boldsymbol{\Sigma}^{-1} \boldsymbol{1}_n \boldsymbol{A}_n' \boldsymbol{\Sigma}^{-1} \boldsymbol{A_n}-\boldsymbol{A}_n' \boldsymbol{\Sigma}^{-1} \boldsymbol{1}_n \boldsymbol{1}_n' \boldsymbol{\Sigma}^{-1} \boldsymbol{A_n}},
\end{equation*}
Applying the Woodbury, Sherman \& Morrison formula, we see \(\boldsymbol{\Sigma}^{-1}\) is a block diagonal matrix of
\begin{equation*}
\boldsymbol{\Sigma_l}^{-1}=\frac{1}{\sigma^2 (1-\rho) (1-\rho +n_l)}\left[ (1+(n-1)\rho) \boldsymbol{I}_{n_l}-\rho \boldsymbol{J}_{n_l}\right].
\end{equation*}
Let \(\bar{Y}_{l(a)}\) denote the sample mean of outcome and \(n_{l(a)}\) number of observations for the control \((a=0)\) and treatment group \((a=1)\) within the \(l^{th}\) subclass, \(l=1,2,...,L\). The treatment effect can be expressed as a weighted sum of two terms
\begin{equation*}
\hat{\tau}=\lambda \hat{\tau}_1+(1-\lambda) \hat{\tau}_0,
\end{equation*}
where \(\lambda =\frac{\rho D1}{\rho D1 +(1-\rho) D2}\), \(\hat{\tau}_1=\frac{C1}{D1}\) and \(\hat{\tau}_0=\frac{C2}{D2}\),
\begin{equation*}
\begin {array}{lr}
    C1=\sum q_l n_l \times \sum q_l n_{l(1)} n_{l(0)} \left( \bar{Y}_{l(1)}-\bar{Y}_{l(0)} \right),\\
    C2=\sum q_k n_{l(0)} \times \sum q_l n_{l(1)} \bar{Y}_{l(1)}- \sum q_l n_{l(1)} \times \sum q_l n_{l(0)} \bar{Y}_{l(0)},\\
    D1=\sum q_l n_l \times \sum q_l n_{l(1)} n_{l(0)} ,\\
    D2=\sum q_l n_{l(1)} \times \sum q_l n_{l(0)},
  \end{array}
\end{equation*}
\(q_l=(1-\rho+\rho n_l)^{-1}\), \(n_l=n_{l(0)}+n_{l(1)}\) and the summations are over \(l=1,...,L\). To gain better insight into this estimator, it should help to consider two special matching cases. 

The first example is a matched twin experiment, where for each treated unit there is a untreated twin. Here, we have a \(2n\times 2n\) block diagonal matrix  \(\boldsymbol{\Sigma_{2n}}=\boldsymbol{I_n \otimes J_2}+\sigma_0 \boldsymbol{I_{2n}}\). Thus, \(\sigma=1+\sigma_0^2\), \(\rho=\frac{1}{1+\sigma_0^2}\), \(n_k=2\), \(n_{k(0)}=n_{k(1)}=1\). Substitute them into the treatment effect formula derived above, we have the same 1:1 matching estimator of treatment effect \(\hat{\tau}=\bar{Y}_1-\bar{Y}_0\). 

The second example is a stratified randomized experiment, where the true propensity of treatment assignment is known. Suppose the strata are equal sized, \(\boldsymbol{\Sigma}\) is a block diagonal matrix of \(\boldsymbol{I_L \otimes J_n}+\sigma_0 \boldsymbol{I_n}\), where \(L\) is total number of strata, the total sample size is \(N=Ln\).  It is straight forward to derive \(\sigma=1+\sigma_0^2\), \(\rho=\frac{1}{1+\sigma_0^2}\), \(n_l=n\), for \(l=1,...,L\). Then the treatment effect is a weighted sum of \( \hat{\tau}_0=\bar{Y}_1-\bar{Y}_0,\) and \(\hat{\tau}_1= \frac{\sum n_{l(0)} n_{l(1)} \left( \bar{Y}_{l(1)}-\bar{Y}_{l(0)} \right)}{\sum n_{l(0)} n_{l(1)}} \).  Where the weight \(\lambda=\frac{N\sum n_{l(0)} n_{l(1)}}{n_1 n_0 \sigma_0^2+N\sum n_{l(0)} n_{l(1)}}\) is a function of sample sizes and \(\sigma_0^2\). We can see when \(\sigma_0^2 \rightarrow 0\), then \(\lambda \rightarrow 1\), \(\tau \rightarrow \hat{\tau}_1\). That is when the outcomes are measured without error, the treatment effect is a weighted average of \( \bar{Y}_{l(1)}-\bar{Y}_{l(0)} \), i.e. the group mean difference for each strata. As \(\sigma_0^2\) increase, \(\lambda\) decrease, then the estimate of \(\tau\) puts more weights on \(\hat{\tau}_0\). In other words, GP estimate of treatment is a shrinkage estimator, where it shrinks the strata level treatment effect more towards the overall sample mean difference when outcome variance is larger. 

More generally, instead of 0/1 match, the sample units may be matched in various degrees. By letting the covariance function takes a squared-exponential form, it offers a way to specify a distance matching, which closely resembles Mahalanobis distance matching. For a pair of "matched" individuals, i.e. sample units with the same set of confounding variables \(\boldsymbol{v_i} = \boldsymbol{v_j} \), the model specifies \(Corr(Y_i^{(0)}, Y_j^{(0)}) = 1 \). In other words, the "matched" individuals are expected to be exchangeable.  As the data points move further apart in the covariate space of \( \Omega_v\), their correlation becomes smaller. When the distant is far apart sufficiently, the model specifies \(Corr(Y_i^{(0)}, Y_j^{(0)}) \approx 0 \) or "unmatched".  Distinct length scale parameters are used to allow for some confounder playing more important roles than others in matching. By manipulating the values of \(v_i\) and the corresponding length scale parameter, one could formulate the SE covariance matrix to reflect the known 0/1 or various degree of matching structure. However, the matching structure is usually unknown, and was left to be estimated in the GPMatch model informed by the observed data.

\subsection{Doubly Robust Estimator of ATE}

\begin{theorem}
Let the true treatment effect be \(\tau^* \), the GPMatch estimator is an unbiased estimate of the average treatment effect, i.e. \(E(\tau_i) = \tau^* \), for \(i = 1, ...n,\) when either one of the condition is true:  i) the GP mean function is correctly specified, i.e. \(  E(Z_i^\prime \hat{\gamma}) =Y_i^{(0)}\); and ii) the GP covariance function is correctly specified, in the sense that, from the weight-space point of view of GP regression, the weighted sum of treatment assignment \(\tilde{A_i}\) correctly specifies the true treatment propensity \( \pi_i = Pr(A_i =1) \).
\end{theorem}

\begin{proof}

It is relatively straight forward for the first part. From the GPMatch model \eqref{eq:d2}   \(\boldsymbol{Y_n} \sim MVN(\boldsymbol{Z^\prime \gamma, \Sigma})\),  when the linear regression model fits the potential outcome correctly, i.e. \(  E(Z_i^\prime \hat{\gamma}) =Y_i^{(0)}\), then \(\Sigma\) degenerate to a diagonal matrix, suggesting all units are exchangeable. It follows \( E(\hat{\tau}) = \tau^* \),  the treatment effect is correctly estimated.   

The second part proceeds as the following. From the weight-space point of view, the  GPMatch model predicts the potential outcomes using a weighted sum of the observed outcomes, 
\begin{equation} 
\begin {array}{lr}
    \hat{Y_i}^{(a)} = \sum_{j=1}^n {w_{ij} (Y_{j}- A_{j} \hat{\tau} )} + a \hat{\tau} = \tilde{Y_i} + (a - \tilde{A_i}) \hat{\tau}, \label{eq:d5a}
  \end{array}
\end{equation}
where \(\tilde{Y_i}=\sum_{j=1}^n w_{ij} Y_j\) and \(\tilde{A}_i=\sum_{j=1}^n w_{ij} A_j\), for \(i=1,...,n\). The weight \( w_{ij} = \frac{\kappa_{ij}}{\sum_j\kappa_{ij}} \) where \(\kappa_{ij}=\boldsymbol{k}(\boldsymbol{v_j})'\boldsymbol{\Sigma}^{-1}\), with \(\boldsymbol{k}(\boldsymbol{v_j})=\left( k(\boldsymbol{v_j},\boldsymbol{v_i}) \right)_{n\times 1}\).  Thus, the \(\tilde{Y_i}\) and \(\tilde{A_i}\) could be considered as the Nadaraya-Watson estimator of the observed outcomes and treatment assignment for each of the i-th unit in the sample.  
The estimate of treatment effect could be obtained by solving \( \frac{\partial{\sum_{i=1}^n \left( Y_i - \hat{Y_i}^{(A_i)} \right)^2}}{\partial{\tau}} = 0 \). We can see that, given a known GP covariance function, the GPMatch treatment effect \(\hat{\tau}\) is an M-estimator that satisfies \(\sum \Psi_i(\hat{\tau}) = 0\), where
\begin{equation} \label{eq:d6}
\begin {array}{lr}
\Psi_i(\tau) = \left( Y_i-\tilde{Y_i}-\tau (A_i-\tilde{A}_i) \right)(A_i-\tilde{A}_i)=0, 
 \end{array}
\end{equation}
Let the true propensity be \(\pi_i = Pr(A_i) \), given the SUTVEA, we have \( Y_i = A_i Y_i^{(1)} + (1-A_i) Y_i^{(0)} + \epsilon_i \). Given the true treatment effect \(\tau^*\), it can be derived that \( Y_i^{(a)} = E(Y_i) + (a-\pi_i) \tau^* \).  When \(\tilde{A_i} = \pi_i \) is true, we have \(\Psi_i(\tau) = [E(Y_i)-\tilde{Y_i} + (A_i-\pi_i) (\tau - \tau^*) +\epsilon_i](A_i - \pi_i) \). Thus, the GPMatch estimator is an M-estimator of ATE, where the estimating function is conditionally unbiased, i.e. \(E(\Psi_i({\tau^*})) = 0\), for \(i = 1, ...n\), when the GP covariance function is correctly specified in the sense \(\tilde{A_i} = \pi_i \) . 

\end{proof}

\begin{remark}

There are several remarks worth noting. First, the equation \eqref{eq:d6} is the empirical correlation of the residuals from the outcome model and the residuals from the propensity of treatment assignment. Thus, GPMatch method attempts to induce independence between  two residuals - one from treatment selection process and one from the outcome modeling, just as the G-estimation equation suggested in \cite{robins2000} and later in \cite{vansteelandt2014}. Unlike the moment based G-estimator, which requires fitting of two separate models for the outcome and propensity score, the GPMatch approach estimates covariance parameters the same time as it estimates the treatment and mean function parameters. All within a full Bayesian framework. 

Second, some data points may have treatment propensity close to 0 or 1. Those data usually are a cause of concern in causal inference. In the naive regression type of model, it may cause unstable estimation without added regularization.  In the  IPTW type of method, a few data points may put undue influence over the estimation of treatment effect. In matching methods, these data points often are discarded. Such practice could lead to sample no longer representative of the target population.  Like the G-estimation,  we see in the equation \eqref{eq:d6}, these data points receive zero or near zero value of \((A_i - \tilde{A_i})\), putting very little influence over the estimation of treatment effect. Thus GPMatch shares the same added robustness as the G-estimation against the lack of overlapping.    

At last, the GPMatch model with a parametric mean function will be predicting the potential outcomes for any new unit.
Given the model setup, two regression surfaces are predicted, where the distance between the two regression surfaces represents the treatment effect.  By including the treatment by covariate interactions, the model could offer conditional treatment effect as a function of the patient characteristics. Although the model specifications presented in section 2.3 suggest using a parametric linear regression equation for modeling  the treatment effect \(\tau(x_i)\), it is always difficult to know if any higher order terms should be included in the model. One may consider introducing a few fixed basis functions instead, estimation of the regression coefficients could inform existence of any nonlinear or heterogeneous treatment effect. 
   
\end{remark}

\section{Simulation Studies}

To empirically evaluate the performances of GPMatch in a real world setting where neither matching structure nor functional form of the outcome model are known, we conducted three sets of simulation studies to evaluate the performances of the GPMatch approach to causal inference. The first set evaluated frequentist performance of GPMatch. The second set compared the performance of GPMatch against MD match, and the last  set utilized the widely used Kang and Schafer design, comparing the performance of GPMatch against some commonly used methods.

In all simulation studies, the GPMatch approach used squared exponential covariate function, including only treatment indicator in the mean and all observed covariates into the covariance function, unless otherwise noted. The results were compared with the following widely used causal inference methods: sub-classification by PS quantile (QNT-PS); AIPTW, linear model with PS adjustment (LM-PS), linear model with spline fit PS adjustment (LM-sp(PS)) and BART. Cubic B-splines with knots based on quantiles of PS were used for LM-sp(PS). We also considered direct linear regression model (LM) as a comparison. The ATE estimates were obtained by averaging over 5000 posterior MCMC draws, after 5,000 burn in.  For each scenario, three sample sizes were considered, \(N\) = 100, 200, and 400.The standard error and the 95\% symmetric interval estimate of ATE for each replicate were calculated from the 5,000 MCMC chain. For comparing performances of different methods, all results were summarized over N=100 replicates by the root mean square error RMSE \(=\sqrt{\sum (\hat{\tau}_i-\tau)^2/N}\), median absolute error MAE \(=median\mid \hat{\tau}_i-\tau\mid\), coverage rate Rc = (the number of intervals that include \(\tau)/N\) of the 95\% symmetric posterior interval, the averaged standard error estimate \(SE_{ave}=\sum \hat{\sigma}_i/N\), where \(\hat{\sigma}_i\) is the square root of the estimated standard deviation of \(\hat{\tau}_i\), and the standard error of ATE was calculated from 100 replicates \(SE_{emp}=\sqrt{\sum (\hat{\tau}_i-\bar{\hat{\tau_i}})^2/(N-1)}\).

\subsection{Well Calibrated Frequentist Performances}
Let the observed covariate \(x \sim N(0,1)\) and the unobserved covariates \(\{U_0, U_1, U_2, \epsilon\} \sim^{iid} N(0,1)\). The potential outcome was generated by \(y^{(a)} = e^x+(1+\gamma_1 U_1)\times a + \gamma_0 U_0\) for \(a=0,1\), where the true treatment effect was \(1+\gamma_1 U_{1i}\) for the i-th individual unit. The \((U_0,U_1)\) are unobserved covariates. The treatment was selected for each individual following \(logit(P(A=1|X))=-0.2+(1.8X)^{1/3} + \gamma_2 U_2^2\). The observed outcome was generated by \(y | x,a = y^{(a)} + \gamma_3 \epsilon \).  Four parameter settings were considered for the  combinations of \(\{\gamma_0, \gamma_1 , \gamma_2, \gamma_3 \}\): \(\{0.5, 0, 0, \sqrt{0.75}\} \), \(\{1, 0.15, 0, 0\} \), \(\{0.5, 0, 0.7, \sqrt{0.75}\} \), and \(\{1, 0.15, 0.7, 0\} \). In the \(1^{st}\) and \(3^{rd}\) settings, let \(\tau_i = 1\). In the \(2^{nd}\) and \(4^{th}\) settings, the treatment effect \(\tau_i \sim (1, \gamma_1^2)\), varying among individual units.  Except for the first setting, the simulation settings included unmeasured confounders  \(U_1\) and/or \(U_2\).    

\begin{figure}[bt]
\centering
\includegraphics[width=14cm]{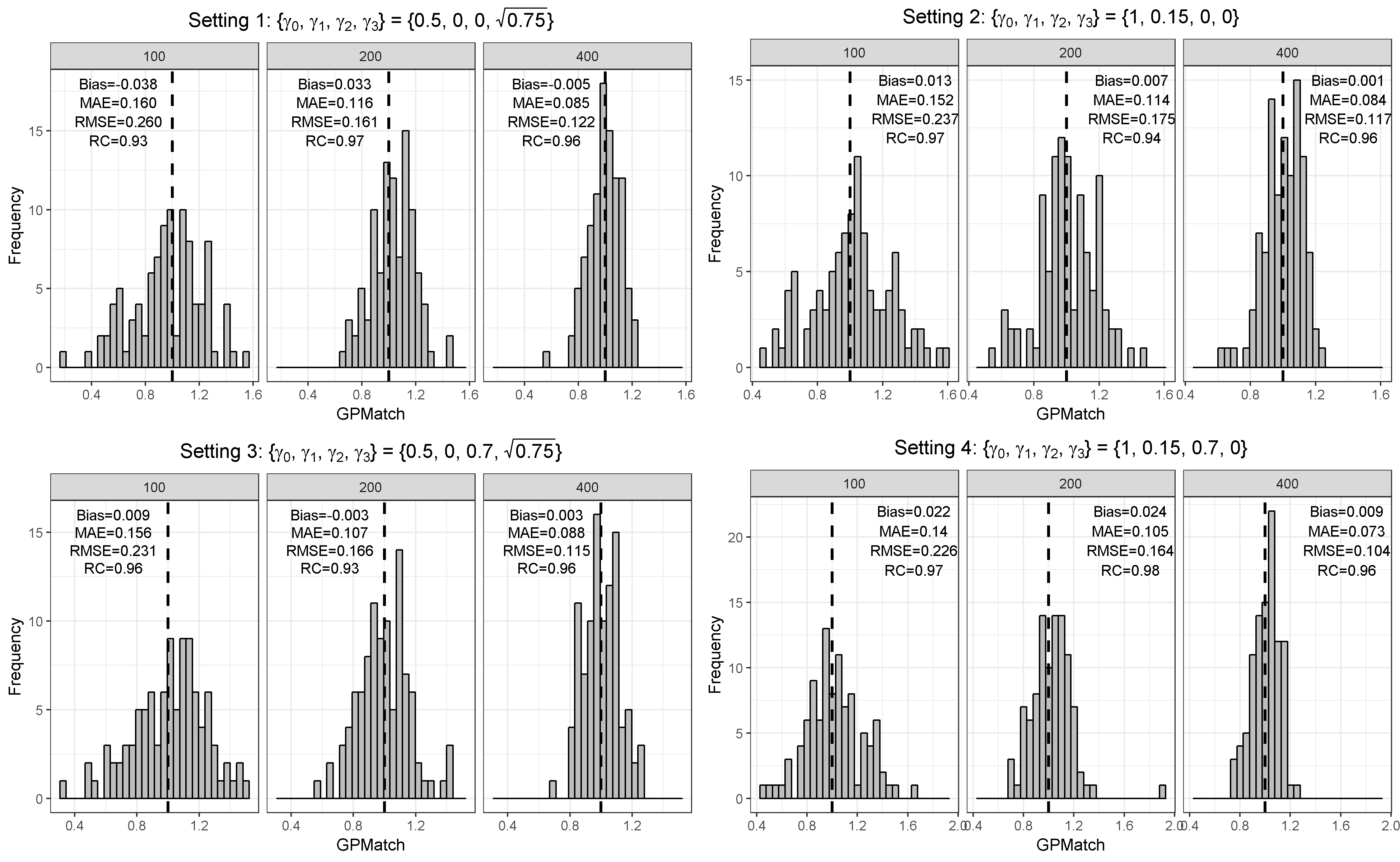}
\caption{Distribution of the GPMatch Estimate of ATE, by Different Sample Sizes under the Single Covariate Simulation Study Settings}
\label{fig:2}
\end{figure}

\rowcolors{1}{}{light-gray}
\begin{longtable}{rccccccc}
\caption{Results of ATE Estimates under the Single Covariate Simulation Study Settings.}\label{table:1}  \\
\hline
\showrowcolors
\rowcolor{light-gray2}\textbf{Method} & \textbf{Sample Size} & \textbf{RMSE} & \textbf{MAE} & \textbf{Bias} & \textbf{Rc} & \textbf{\(\boldsymbol{SE_{avg}}\)} & \textbf{\(\boldsymbol{SE_{emp}}\)}\\
\hline
\endfirsthead
\multicolumn{8}{c}%
{\tablename\ \thetable\ -- \textit{Continued from previous page}} \\
\hline
\rowcolor{light-gray2}\textbf{Method} & \textbf{Sample Size} & \textbf{RMSE} & \textbf{MAE} & \textbf{Bias} & \textbf{Rc} & \textbf{\(\boldsymbol{SE_{avg}}\)} & \textbf{\(\boldsymbol{SE_{emp}}\)}\\
\hline
\endhead
\hline \multicolumn{8}{r}{\textit{Continued on next page}} \\
\endfoot
\endlastfoot
\multicolumn{8}{c}{Setting 1: \(\{\gamma_0, \gamma_1 , \gamma_2, \gamma_3 \}=\{0.5, 0, 0, \sqrt{0.75}\}\)}\\
Gold & 100 & 0.243 & 0.165 & -0.066 & 0.93 & 0.216 & 0.235\\
  & 200 & 0.149 & 0.109 & 0.027 & 0.94 & 0.150 & 0.147\\
  & 400 & 0.123 & 0.087 & -0.007 & 0.93 & 0.107 & 0.123\\
 &  &  &  &  &  &  & \\
  GPMatch & 100 & 0.260 & 0.160 & -0.038 & 0.93 & 0.242 & 0.258\\
  & 200 & 0.161 & 0.116 & 0.033 & 0.97 & 0.167 & 0.159\\
  & 400 & 0.122 & 0.085 & -0.005 & 0.96 & 0.118 & 0.123\\
\multicolumn{8}{c}{Setting 2: \(\{\gamma_0, \gamma_1 , \gamma_2, \gamma_3 \}=\{1, 0.15, 0, 0\}\)}\\
Gold & 100 & 0.220 & 0.134 & -0.011 & 0.92 & 0.213 & 0.221\\
  & 200 & 0.159 & 0.098 & 0.001 & 0.94 & 0.151 & 0.159\\
  & 400 & 0.107 & 0.077 & -0.003 & 0.95 & 0.107 & 0.108\\
 &  &  &  &  &  &  & \\
GPMatch & 100 & 0.237 & 0.152 & 0.013 & 0.97 & 0.244 & 0.238\\
  & 200 & 0.175 & 0.114 & 0.007 & 0.94 & 0.169 & 0.175\\
  & 400 & 0.117 & 0.084 & 0.001 & 0.96 & 0.117 & 0.118\\
\multicolumn{8}{c}{Setting 3: \(\{\gamma_0, \gamma_1 , \gamma_2, \gamma_3 \}=\{0.5, 0, 0.7, \sqrt{0.75}\}\)}\\
Gold & 100 & 0.228 & 0.137 & -0.016 & 0.92 & 0.214 & 0.228\\
  & 200 & 0.154 & 0.099 & 0.005 & 0.94 & 0.151 & 0.155\\
  & 400 & 0.113 & 0.078 & 0.001 & 0.94 & 0.107 & 0.114\\
 &  &  &  &  &  &  & \\
GPMatch & 100 & 0.231 & 0.156 & 0.009 & 0.96 & 0.237 & 0.232\\
  & 200 & 0.166 & 0.107 & -0.003 & 0.93 & 0.164 & 0.167\\
  & 400 & 0.115 & 0.088 & 0.003 & 0.96 & 0.114 & 0.115\\
\multicolumn{8}{c}{Setting 4: \(\{\gamma_0, \gamma_1 , \gamma_2, \gamma_3 \}=\{1, 0.15, 0.7, 0\}\)}\\
Gold & 100 & 0.209 & 0.148 & 0.015 & 0.96 & 0.215 & 0.209\\
  & 200 & 0.136 & 0.098 & 0.008 & 0.97 & 0.152 & 0.136\\
  & 400 & 0.095 & 0.076 & -0.002 & 0.98 & 0.107 & 0.095\\
 &  &  &  &  &  &  & \\
GPMatch & 100 & 0.226 & 0.140 & 0.022 & 0.97 & 0.238 & 0.226\\
  & 200 & 0.164 & 0.105 & 0.024 & 0.98 & 0.169 & 0.163\\
  & 400 & 0.104 & 0.073 & 0.009 & 0.96 & 0.114 & 0.104\\  
\hline  
\end{longtable}

\begin{tablenotes}\footnotesize
\item RMSE = root mean square error; MAE = median absolute error; Bias = Estimate-True; Rc = Rate of coverage by the 95\% interval estimate; \(SE_{avg}\) = average of standard error estimate from all replicate; \(SE_{emp}\) = standard error of ATE estimates from all replicate;
\item Gold: Using the true outcome generating model;
\item GPMatch: Bayesian marginal structural model with Gaussian process prior, only treatment effect is included in the mean function; covariance function includes \(X\). 
\bigskip
\end{tablenotes}
\rowcolors{1}{}{}

The simulation results were summarized in the histogram of the posterior mean over the 100 replicates across three sample sizes in Figure ~\ref{fig:2}. Table \ref{table:1} presented the results of GPMatch and the gold standard.  The gold standard was obtained by fitting the true outcome generating model.  Under all settings, GPMatch presented well calibrated frequentist properties with nominal coverage rate, and only slightly larger RMSE. The averaged bias, RMSE and MAE quickly improve as sample size increases, and perform as well as the gold standard with the sample size of 400. Comparison of the RMSE and MAE with the results using other causal inference methods were presented in Figures S1 - S4.

\subsection{Compared to Manhalanobis Distance Matching}
To compare the performances between the MD matching and GPMatch, we considered a simulation study with two independent covariates \(x_1\), \(x_2\) from the uniform distribution \(U(-2,2)\), treatment was assigned by letting \(A_i\sim Ber(\pi_i)\), where 
\begin{equation*}
logit \pi_i=-x_1-x_2.
\end{equation*}
The potential outcomes were generated by 
\begin{equation*}
\begin{array}{cr}
y_i^{(a)} = 3+5a+x_{1i}^3,\\
Y_i|X_i, A_i \sim N(y_i^{(A_i)}, 1).
\end{array}
\end{equation*}
The true treatment effect is 5. Three different sample sizes were considered N= 100, 200 and 400. For each setting, 100 replicates were performed and the results were summarized.

\begin{figure}[bt]
\centering
\includegraphics[width=10cm]{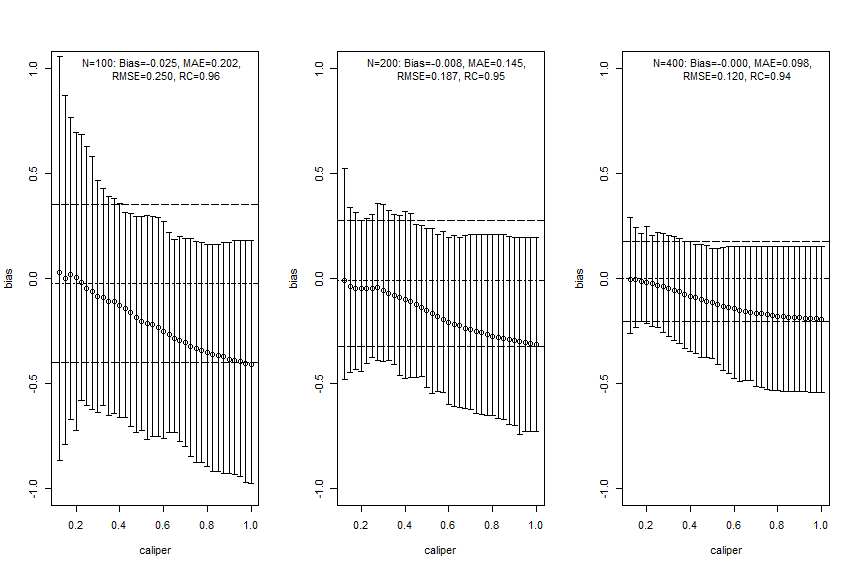}
\caption{Simulation Study Results of comparing GPMatch with Manhalanobis Distance Matching Methods. The circles are the averaged biases of estimates of ATE using Mahalanobis matching with corresponding calipers. The corresponding vertical lines indicate the ranges between 5th and 95th percentiles of the biases. The horizontal lines are the averaged ATE (short dashed line), and the 5th percentile and 95th percentile (long dashed line) of the biases of the estimates from GPMatch.}
\label{fig:1}
\end{figure}

We estimated ATE by applying Mahalanobis distance matching and GPMatch. The MD matching considered caliper varied from 0.125 to 1 with step size 0.025, including both \(X_1\) and \(X_2\)  in the matching using the function Match in R package Matching by \cite{sekhon2007multivariate}. The averaged bias and its 95\%-tile and 5\%-tile were presented as vertical lines corresponding to different calipers in Figure ~\ref{fig:1}. To be directly comparable to the matching approach, the GPMatch estimated the ATE by including treatment effect only in modeling the mean function, both \(X_1\) and \(X_2\) were considered in the covariance function modeling. The posterior results were generated with 5,000 MCMC sample after 5,000 burn-in. Its averaged bias (short dashed horizontal line) and 5\% and 95\%-tiles of the ATE estimate (long dashed horizontal lines) were presented on the Figure ~\ref{fig:1} for each the sample sizes. Also presented in the Figure were the bias, median absolute error (MAE), root mean square error (RMSE), and rate of coverage rate (Rc) summarized over 100 replicates of GPMatch. The bias from the matching method increases with caliper; the width of interval estimate varies by sample size and caliper. It reduces with increased caliper for the sample size of 100, but increases with increased caliper for sample size of 400. In contrast, GPMatch produced a much more accurate and efficient estimate of ATE for all sample sizes, with unbiased ATE estimate and nominal coverage rate. The 5\% and 95\%-tiles of ATE estimates are always smaller than those from the matching methods for all settings considered, suggesting better efficiency of GPMatch.

\subsection{Performance under Dual Misspecification }
Following the well-known simulation design suggested by \cite{Kang2007}, covariates \(z_1,z_2,z_3,z_4\) were independently generated from the standard normal distribution \(N(0,1)\). Treatment was assigned by \(A_i  \sim Ber(\pi_i)\), where
\begin{equation*}
logit \pi_i  = -z_{i1}  + 0.5z_{i2}  - 0.25z_{i3}  - 0.1z_{i4}.
\end{equation*}
The potential outcomes were generated for \(a=0,1\) by
\begin{equation*}
\begin{array}{cl}
y_i^{(a)} = 210+5a+27.4z_{i1}+13.7z_{i2}+13.7z_{i3}+13.7z_{i4},\\
Y_i|A_i, X_i \sim N(y^{(A_i)}, 1).
\end{array}
\end{equation*}
The true treatment effect is 5. To assess the performances of the methods under the dual miss-specifications, the transformed covariates \(x_1  = exp(z_1/2),x_2  = z_2/(1 + exp(z_1)) + 10,x_3  = \left(\frac{z_1 z_3}{25}+ 0.6\right)^3\), and \(x_4  = (z_2  + z_4  + 20)^2 \) were used in the model instead of \(z_i\).    

Two GPMatch models were considered: GPMatch1 modeled the treatment effect only and GPMatch2 modeled all four covariates \(X_1-X_4\) in the mean function model. Both included \(X_1-X_4\)  with four distinct length scale parameters. The PS was estimated using two approaches including the logistic regression model on \(X_1-X_4\) and the covariate balancing propensity score method  (CBPS, \cite{imai2014}) applied to \(X_1-X_4\).  The results corresponding to both versions of PS were presented. Summaries over all replicates were presented in Table \ref{table:2}, and the RMSE and the MAE were plotted in Figure ~\ref{fig:3}, for all methods considered. As a comparison, the gold standard which uses the true outcome generating model of \(Y \sim Z_1-Z_4\) was also presented.  Both GPMatch1 and GPMatch2 clearly outperforms all the other causal inference methods in terms of bias, RMSE, MAE, Rc, and the \(SE_{ave}\) is closely matched to \(SE_{emp}\).  The ATE and the corresponding SE estimates improve quickly as sample size increases for GPMatch. In contrast, the QNT\_PS, AIPT, LM\_PS and LM\_sp(PS) methods show little improvement over increased sample size, so is the simple LM. Improvements in the performance of GPMatch over existing methods are clearly evident, with more than 5 times accuracy in RMSE and MAE compared to all the other methods except for BART. Even compared to the BART results, the improvement in MAE is nearly twice for GPMatch2, and about 1.5 times for the GPMatch1. Similar results are evident in RMSE and averaged bias.  The lower than nominal coverage rate is mainly driven by the remaining bias, which quickly reduces as sample size increases.  Additional results are presented in Figure S5.

\rowcolors{1}{}{light-gray}
\begin{longtable}{rccccccc}
\caption{Results of ATE Estimates using Different Methods under the Kang and Shafer Dual Misspecification setting.}\label{table:2}  \\
\hline
\showrowcolors
\rowcolor{light-gray2}\textbf{Method} & \textbf{Sample Size} & \textbf{RMSE} & \textbf{MAE} & \textbf{Bias} & \textbf{Rc} & \textbf{\(\boldsymbol{SE_{avg}}\)} & \textbf{\(\boldsymbol{SE_{emp}}\)}\\
\hline
\endfirsthead
\multicolumn{8}{c}%
{\tablename\ \thetable\ -- \textit{Continued from previous page}} \\
\hline
\rowcolor{light-gray2}\textbf{Method} & \textbf{Sample Size} & \textbf{RMSE} & \textbf{MAE} & \textbf{Bias} & \textbf{Rc} & \textbf{\(\boldsymbol{SE_{avg}}\)} & \textbf{\(\boldsymbol{SE_{emp}}\)}\\
\hline
\endhead
\hline \multicolumn{8}{r}{\textit{Continued on next page}} \\
\endfoot
\endlastfoot

Gold & 100 & 0.224 & 0.150 & 0.011 & 0.95 & 0.225 & 0.225\\
 & 200 & 0.171 & 0.125 & -0.015 & 0.94 & 0.163 & 0.171\\
 & 400 & 0.102 & 0.063 & -0.015 & 0.96 & 0.112 & 0.102\\
GPMatch1 & 100 & 2.400 & 1.606 & -1.254 & 0.92 & 2.158 & 2.057\\
 & 200 & 1.663 & 1.309 & -1.051 & 0.86 & 1.213 & 1.295\\
 & 400 & 0.897 & 0.587 & -0.564 & 0.86 & 0.673 & 0.701\\
GPMatch2 & 100 & 1.977 & 1.358 & -0.940 & 0.91 & 1.672 & 1.748\\
 & 200 & 1.375 & 1.083 & -0.809 & 0.82 & 0.980 & 1.117\\
 & 400 & 0.761 & 0.484 & -0.432 & 0.87 & 0.567 & 0.629\\
QNT\_PS\(^a\) & 100 & 7.574 & 6.483 & -6.234 & 0.970 & 7.641 & 4.324\\
 & 200 & 7.408 & 6.559 & -6.615 & 0.860 & 5.199 & 3.353\\
 & 400 & 7.142 & 6.907 & -6.797 & 0.500 & 3.576 & 2.203\\
QNT\_PS\(^b\) & 100 & 8.589 & 7.360 & -7.177 & 0.970 & 7.541 & 4.744\\
 & 200 & 8.713 & 8.121 & -7.964 & 0.720 & 5.214 & 3.550\\
 & 400 & 8.909 & 7.980 & -8.399 & 0.300 & 3.607 & 2.987\\
LM & 100 & 6.442 & 5.183 & -5.556 & 0.65 & 3.571 & 3.277\\
 & 200 & 6.906 & 6.226 & -6.375 & 0.28 & 2.547 & 2.668\\
 & 400 & 7.005 & 6.649 & -6.702 & 0.04 & 1.796 & 2.048\\
AIPTW\(^a\) & 100 & 5.927 & 4.402 & -4.330 & 0.72 & 3.736 & 4.067\\
 & 200 & 19.226 & 5.262 & -7.270 & 0.59 & 4.874 & 17.888\\
 & 400 & 29.405 & 5.603 & -9.676 & 0.36 & 6.115 & 27.908\\
AIPTW\(^b\) & 100 & 5.410 & 4.243 & -3.659 & 0.77 & 3.780 & 4.005\\
 & 200 & 5.780 & 5.075 & -4.950 & 0.52 & 2.712 & 2.999\\
 & 400 & 6.204 & 5.482 & -5.652 & 0.24 & 2.105 & 2.569\\
LM\_PS\(^a\)  & 100 & 5.103 & 3.832 & -4.091 & 0.74 & 3.420 & 3.066\\
 & 200 & 5.392 & 4.648 & -4.793 & 0.53 & 2.452 & 2.483\\
 & 400 & 5.091 & 5.128 & -4.787 & 0.19 & 1.706 & 1.741\\
LM\_PS\(^b\)  & 100 & 5.451 & 4.156 & -4.528 & 0.72 & 3.427 & 3.051\\
 & 200 & 5.891 & 4.981 & -5.278 & 0.46 & 2.466 & 2.631\\
 & 400 & 5.585 & 5.452 & -5.272 & 0.13 & 1.726 & 1.852\\
LM\_sp(PS)\(^a\) & 100 & 4.809 & 3.161 & -3.598 & 0.79 & 3.165 & 3.207\\
 & 200 & 4.982 & 4.152 & -4.266 & 0.52 & 2.250 & 2.587\\
 & 400 & 4.470 & 4.038 & -4.127 & 0.23 & 1.559 & 1.727\\
LM\_sp(PS)\(^b\) & 100 & 4.984 & 3.619 & -3.806 & 0.77 & 3.095 & 3.233\\
 & 200 & 5.237 & 4.374 & -4.507 & 0.51 & 2.248 & 2.681\\
 & 400 & 4.856 & 4.484 & -4.494 & 0.18 & 1.585 & 1.851\\
BART & 100 & 3.148 & 2.504 & -2.491 & 0.79 & 2.163 & 1.935\\
 & 200 & 2.176 & 1.870 & -1.726 & 0.74 & 1.308 & 1.332\\
  & 400 & 1.283 & 0.942 & -0.997 & 0.71 & 0.757 & 0.812\\
\hline  
\end{longtable}

\begin{tablenotes}\footnotesize
\item \(^a\) Propensity score estimated using logistic regression on \(X_1-X_4\).
\item \(^b\) Propensity score estimated using CBPS on \(X_1-X_4\).
\item RMSE = root mean square error; MAE = median absolute error; Bias = Estimate-True; Rc = Rate of coverage by the 95\% interval estimate; \(SE_{avg}\) = average of standard error estimate from all replicate; \(SE_{emp}\) = standard error of ATE estimates from all replicate;
\item GPMatch1-2: Bayesian structural model with Gaussian process prior. GPMatch1 including only treatment effect, and GPMatch2 including both treatment effect and \(X_1-X_4\) in the mean function; both including \(X_1-X_4\) in the covariance function. 
\item QNT\_PS: Propensity score sub-classification by quintiles.
\item AIPTW: augmented inversed probability of treatment weighting;  
\item LM: linear regression modeling \(Y \sim X_1-X_4\); 
\item LM\_PS: linear regression modeling with propensity score adjustment.  
\item LM\_sp(PS): linear regression modeling with spline fit propensity score adjustment.
\item BART: Bayesian additive regression tree. 
\bigskip
\end{tablenotes}

\begin{figure}[!htb]
\centering
\includegraphics[width=13cm]{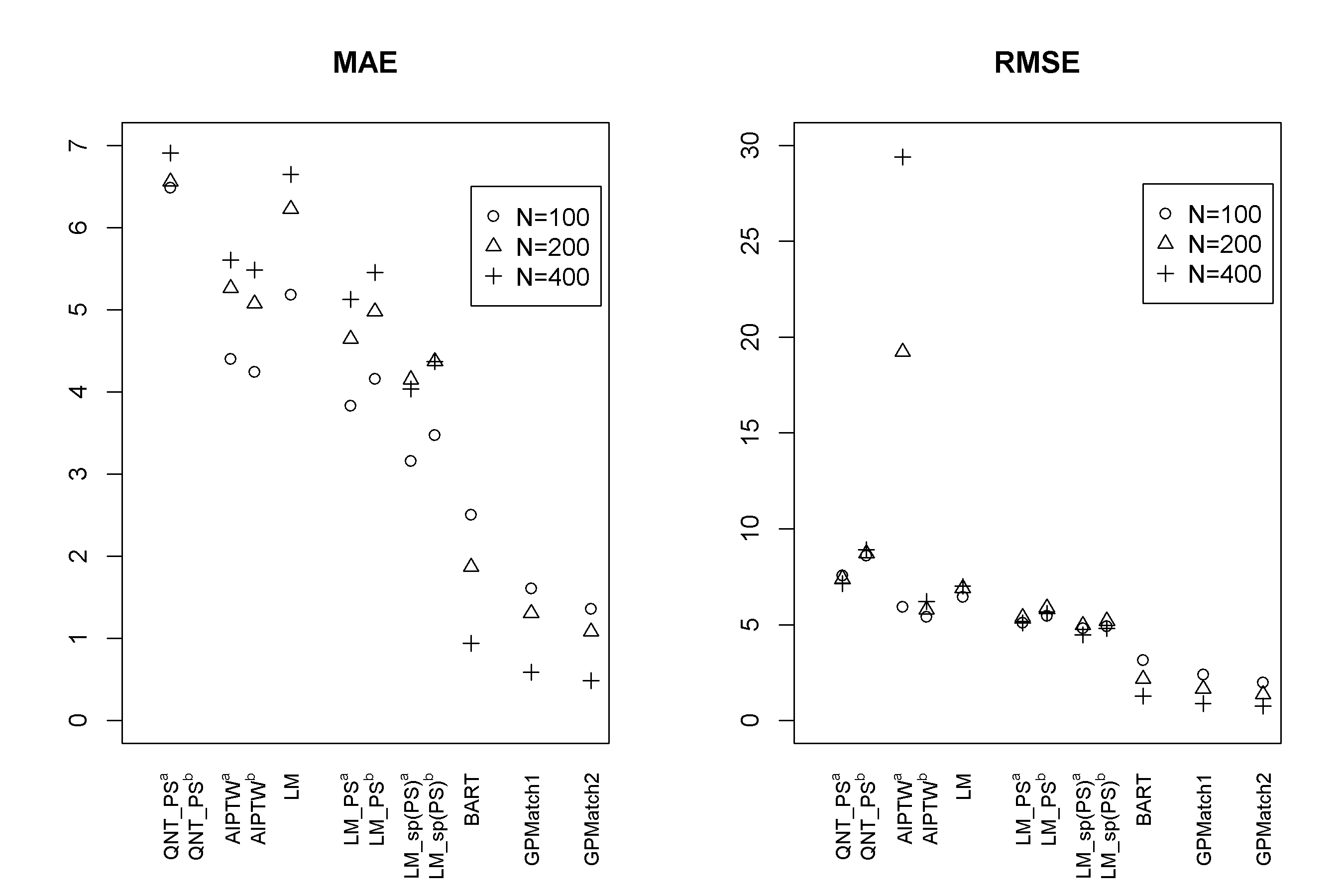}
\caption{The RMSE and MAE of ATE Estimates using Different Methods under the Kang and Shafer Simulation Study Setting. GPMatch1-2: Bayesian structural model with Gaussian Process prior. GPMatch1 including only treatment effect, and GPMatch2 including both treatment effect and \(X_1-X_4\) in the mean function; and \(X_1-X_4\) are included in the covariance function. QNT\_PS: Propensity score sub-classification by quintiles. AIPTW: augmented inverse probability of treatment weighting; LM: linear regression modeling \(Y \sim X_1-X_4\); LM\_PS: linear regression modeling with propensity score adjustment. LM\_sp(PS): linear regression modeling with spline fit propensity score adjustment}
\label{fig:3}
\end{figure}

\section{A Case Study}
JIA is a chronic inflammatory disease, the most common autoimmune disease affecting the musculoskeletal organ system, and a major cause of childhood disability. The disease is relatively rare, with an estimated incidence rate of 12 per 100,000 child-year (\cite{harrold2013}). There are many treatment options. Currently, the two common approaches are the non-biologic disease modifying anti-rheumatic drugs (DMARDs) and the biologic DMARDs. Limited clinical evidence suggest that early aggressive use of biologic DMARDs may be more effective (\cite{wallace2014}). Utilizing data collected from a completed prospectively followed up inception cohort research study (\cite{seid2014}), a retrospective chart review collected medication prescription records for study participants captured in the electronic health record system. This comparative study is aimed at understanding whether therapy using early aggressive combination of non-biologic and biologic DMARDs is more effective than the more commonly adopted non-biologic DMARDs monotherapy in treating children with recently (<6 months) diagnosed polyarticular course of JIA. The study is approved by the investigator's institutional IRB. 

\begin{figure}[!htb]
\centering
\includegraphics[width=12cm]{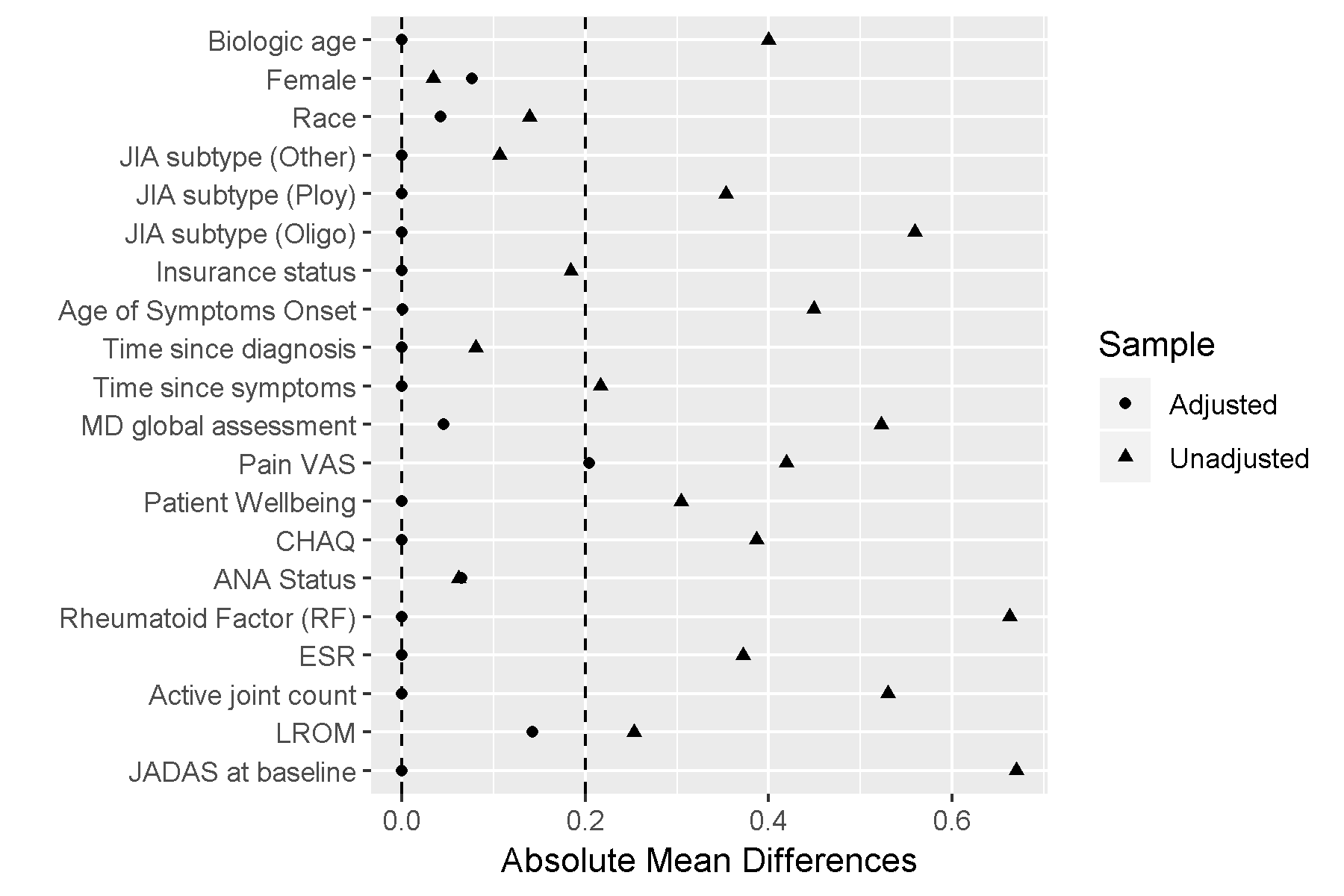}
\caption{Balance Check Results for the Cases Study}
\label{fig:4}
\end{figure}

The primary outcome is the Juvenile Arthritis Disease Activity Score (JADAS) after 6 months of treatment, a disease severity score calculated as the sum of four core clinical measures: physician's global assessment of disease activity (0-10), patient's self-assessment of overall wellbeing (0-10), erythrocyte sedimentation rate (ESR, standardized to 0-10), and number of active joint counts (AJC, truncated to 0-10). It ranges from 0 to 40, with 0 indicating no disease activity. Out of the 75 patients receiving either non-biological or the early combination of biological and non-biological DMARDs at baseline, 52 patients were treated by the non-biologic DMARDs and 23 were treated by the early aggressive combination DMARDs. The patients with longer disease duration, positive rheumatoid factor (RF) presence, higher pain visual analog scale (VAS) and lower baseline functional ability as measured by the childhood health assessment questionnaire (CHAQ), higher lost range of motion (LROM) and JADAS score are more likely to receive the biologic DMARDs prescription. The propensity score was derived using the CBPS method applied to the pre-determined important baseline confounders. The derived PS were able to achieve a desired covariate balance within the 0.2 absolute standardized mean difference (Figure ~\ref{fig:4}), and comparable distributions in important confounders (Figure S6).  
\begin{figure}[!htb]
\centering
\includegraphics[width=12cm]{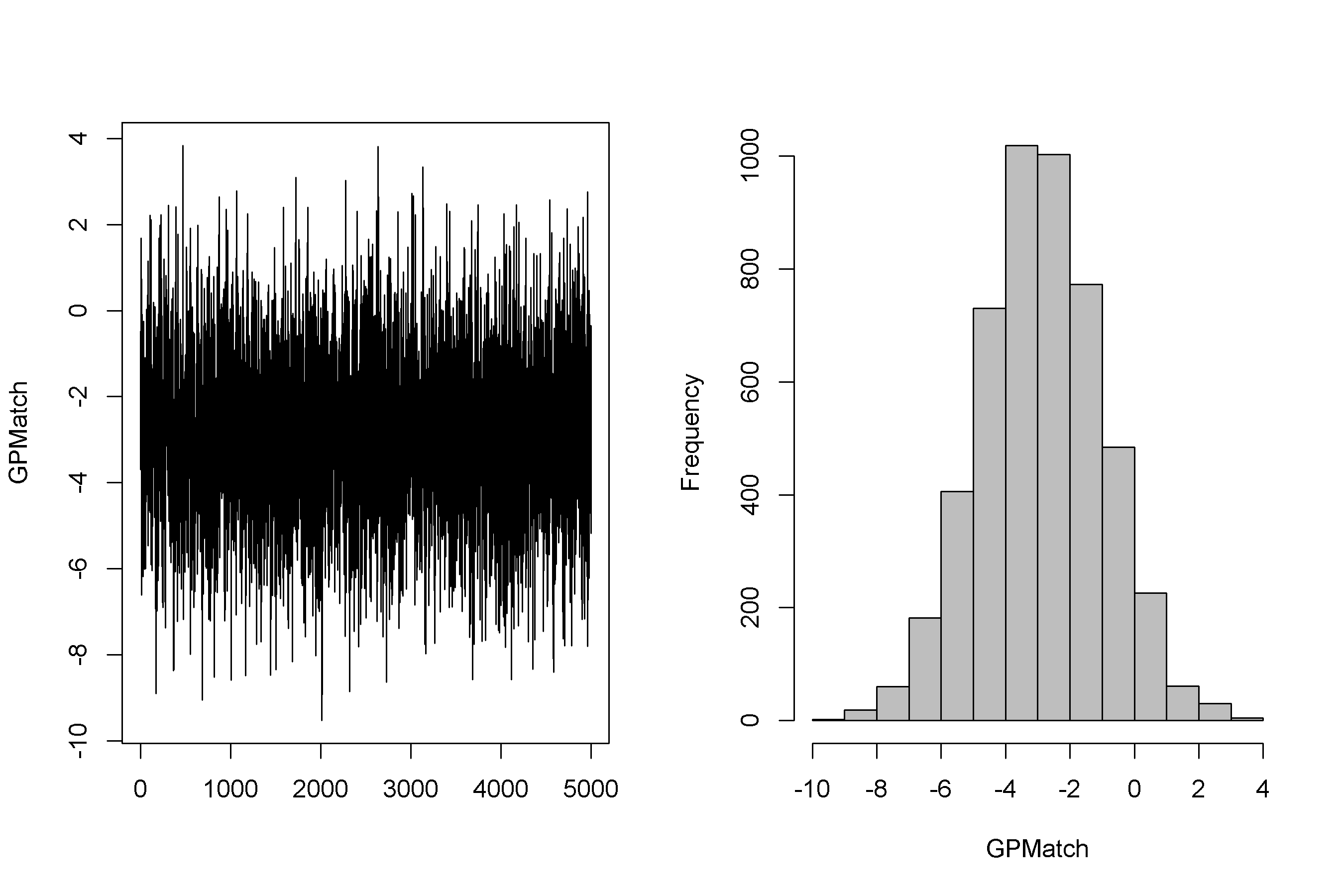}
\caption{Case Study Trace Plot and Histogram}
\label{fig:5}
\end{figure}

\begin{table}[bt]
\caption{Results of Case Study ATE Estimates with None-Matching Methods}
\label{table:3}
\begin{center}
\begin{threeparttable}

\begin{tabular}{rccccc}
\headrow
\thead{Method} & \thead{Estimate} & \thead{SD} & \thead{LL} & \thead{UL}\\
Naïve & -0.338 & 1.973 & -4.205 & 3.529\\
QNT\_PS & -0.265 & 0.792 & -1.817 & 1.286\\
AIPTW & -0.639 & 2.784 & -6.094 & 4.817\\
LM & -2.550 & 1.981 & -6.432 & 1.332\\
LM\_PS & -2.844 & 2.002 & -6.767 & 1.079\\
LM\_sp(PS) & -1.664 & 2.159 & -5.896 & 2.568\\
BART & -2.092 & 1.629 & -5.282 & 1.155\\
GPMatch & -2.902 & 1.912 & -6.650 & 0.789\\
\hline  

\end{tabular}

\begin{tablenotes} \footnotesize
\item SD = standard deviation; LL = lower limit; UL=upper limit; 
\item Naïve: Student-T two group comparisons; 
\item QNT\_PS: Propensity score sub-classification by quintiles.
\item AIPTW: augmented inversed probability of treatment weighting;  
\item LM: linear regression modeling \(Y \sim X\); 
\item LM\_PS: linear regression modeling with propensity score adjustment.  
\item LM\_sp(PS): linear regression modeling with spline fit propensity score adjustment;
\item BART: Bayesian additive regression tree; 
\item GPMatch: Bayesian structural model with Gaussian process prior.
\end{tablenotes}
\end{threeparttable}
\end{center}
\end{table}

The GPMatch model included the baseline JADAS, CHAQ, time since diagnosis at baseline, and time interval between baseline and the six month follow-up visit in modelling the covariance function. These four covariates, along with the binary treatment indicator and an indicator of positive test of rheumatoid factor were used in the partially linear mean function part of the GPMatch. Applying the proposed method, GPMatch obtained the average treatment effect of -2.90 with standard error of 1.91, and the 95\% credible interval of (-6.65, 0.79). Figure ~\ref{fig:5} presents the trace plot and histogram of the posterior distribution of the ATE estimate. The results suggest that, the early aggressive combination of non-biologic and biologic DMARDs as the first line of treatment is more effective, leading to a nearly 3 point of reduction in JADAS six months after treatment, compared to the non-biologic DMARDs treatment to children with a newly diagnosed disease. The results of ATE estimates by GPMatch, naive two group comparison and other existing causal inference methods are presented in Table \ref{table:3}. The LM, LM\_PS, LM\_sp(PS) and AIPTW include the same five covariates in the model along with the treatment indicator. BART used the treatment indicator and those covariates. While all results suggested effectiveness of an early aggressive use of biological DMARD, the naive, PS sub-classificiton by quintiles, and AIPTW suggested a much smaller ATE effect. The BART and PS adjusted linear regression produced results that were closer to the GPMatch results suggesting 2 or 3 points reduction in the JADAS score if treated by the early aggressive combination DMARDs therapy. None of the results were statistically significant at the 2-sided 0.05 level. 

We also applied the covariate matching method to the same dataset based on the same four baseline covariates. Table \ref{table:4} presents the results from using different caliper. As expected, as calipers narrow, the number of observations being discarded increases. Since only 10 patients had RF positive, thus, when the calipers were set to 1 or smaller, we cannot matching on RF positive anymore. Thus, for calipers smaller than 1, all subjects with positive RF were being excluded. When calipers were set at 0.5, about 50\% observations were discarded. When the calipers were set at 0.2, 62 out of 73 observations were discarded, rendering the results obtained from 11 observations only! The estimate of ATE was sensitive to the choices of calipers, ranged from -6.59 to -3.12, making it difficult to interpret the study results.

\begin{table}[bt]
\caption{Results of Case Study ATE Estimates with Matching Method in Case Study}
\label{table:4}
\begin{center}
\begin{threeparttable}

\begin{tabular}{rccccccc}
\headrow
\thead{Caliper} & \thead{} & \thead{2} & \thead{1} & \thead{0.8} & \thead{0.5} & \thead{0.4} & \thead{0.2}\\
ATE & & -3.117 & -4.043 & -4.035 & -5.577 & -6.592 & -3.864\\
SE & &  2.232 & 2.075 & 1.701 & 1.459 & 1.092 & 0.536\\
\# of obs dropped & & 1 & 10 & 21 & 34 & 49 & 62\\
& \multicolumn{7}{c}{Standardized mean difference between two treatment groups}\\
& \textbf{Before Match} & \multicolumn{6}{c}{\textbf{After Match}}\\
JADAS0 & 0.675 &  0.215 & 0.078 & 0.079 & 0.035 & 0.079 & -0.090\\
Time diagnosed & 0.233 &  0.013 & 0.020 & -0.006 & -0.010 & 0.041 & 0.048\\
CHAQ & 0.281 &  0.083 & 0.079 & 0.072 & 0.079 & -0.054 & -0.057\\
RF positive & 0.643 &  0.000 & 0.000 & NA* & NA* & NA* & NA*\\

\hline  

\end{tabular}

\begin{tablenotes} \footnotesize
\item Note: * When the caliper is less than 1, all of the observations with positive RF are excluded. 
\end{tablenotes}
\end{threeparttable}
\end{center}
\end{table}

\section{Conclusions and Discussions}
Bayesian approaches to causal inference commonly consider it as a missing data problem. However, as suggested in \cite{Ding2018}, the causal inference presents additional challenges that are unique in itself than the missing data alone.  Approaches not carefully address these unique challenges are vulnerable to model mis-specifications and could lead to seriously biased results. When not considering the treatment-by-indication confounding,  naive Bayesian regression approaches  could suffer from "regularity induced bias" (\cite{Hahn2018} ). Because no more than one potential outcome could be observed for a given individual unit, the correlation of \( (Y_i^{(1)} , Y_i^{(0)}) \)  is not directly identifiable, leading to "inferential quandary" as suggested in \cite{Dawid2000} .  Extensive simulations presented in \cite{Kang2007, Gutman2017, Hahn2018b} suggested poor operational characteristics observed in many widely adopted causal inference methods.

The proposed GPMatch method offers a full Bayesian causal inference approach that can effectively address the unique challenges inherent in causal inference. First, utilizing GP prior covariance function to model covariance of observed data, GPMatch could estimate the missing potential outcomes much like the matching method. Yet, it avoids pitfalls of many matching methods. No data is discarded, and no arbitrary caliper is required. Instead, the model allows the data to speak by itself via estimating length scale and variance parameters.  The SE covariance function of GP prior offers an alternative distance metric, which closely resembles Mahalanobis distance. It matches data points by the degree of matching proportional to the SE distance, without requiring specification of caliper. For this reason, the GPMatch could utilize data information better than matching procedure. Different length scale parameters are considered for different covariates used in defining SE covariance function. This allows the data to select the most important covariates to be matched on, and acknowledge some variable is more important than others. While the idea of using GP prior for Bayesian causal inference is not new. Utilizing GP covariance function as a matching device is a unique contribution of this study. The matching utility of GP covaraince function is presented analytically  by considering a setting when matching structure is known. We show that GPMatch enjoys doubly robust properties, in the sense that it correctly estimate the averaged treatment effect when either one of the conditions is true: 1) the mean function of the GPMatch correctly specifies the prognostic function of the potential outcome \(Y^{(0)}\); and 2) the GP prior covariance function correctly specifies matching structure.  We show that GPMatch estimates the treatment effect by inducing independence between two residuals: the residual from treatment propensity estimate and the residual from the outcome estimate, much like  the G-estimation method. Unlike the two-staged G-estimation, the estimations of the parameters in covariance function and the mean function for the GPMatch are performed simultaneously.  Therefor, GPMatch regression approach can integrate the benefits of the regression model and matching method and offers a natural way for Bayesian causal inference to address challenges unique to the causal inference problems. The robust and efficient proprieties of GPMatch are well supported by the simulation results designed  to reflect the most realistic settings, i.e. no knowledge of matching or functional form of outcome model is available. 

The validity of  causal inference by GPMatch rests on a weaker version of causal assumptions depicted in the Fig. 1 DAG.  Despite the fact that previous literature have questioned the SUTVA assumption (see stochastic consistency suggested by \cite{Cole2009} and \cite{VanderWeele2009}, and treatment variation discussed in \cite{Rubin1978}), no approach to our knowledge has explicitly acknowledged it as such. Rather, most of the methods imposing an overly rigid assumption that the treatment from the real world as having exactly the same meaning as those from the randomized and tightly controlled experiments, and that the observed outcome is an exact copy of the corresponding potential outcome. Here, our causal assumptions reflect more realistic setting that outcome could be measured with error, and the treatment received by different individuals may vary, even though the treatment prescribed is identical. The assumption of the ignorable treatment assignment is often used exchangeably with the assumption of no unmeasured confounder in the currently literature. The ignorable treatment assignment assumption is necessary to ensure the validity of causal inference obtained from the observed data.  In Fig 1. DAG, we show that unmeasured confounders are  admissible under the ignorable treatment assumption. Specifically, it allows for existence of \(U_1\) and \(U_2\) (both correlated with \(A\) and \(Y\)), as long as they do not open up the back-door path from Y to \(A\) conditional of the measured covariates.  In other words, the causal effect can be identified without bias if we could observe a minimum sufficient set that block the back-door path from \(Y\) to \(A\). Thus, it presents a weaker assumption than the assumption of no unmeasured confounder.  In case any potential violation of the causal assumptions is suspected, external information is needed and GPMatch can be extended to incoporate such uncertainty.  With a weaker version of causal assumptions and by explicitly modeling the mean and covariance functions, the GPMatch is more capable of defending against potential model misspecification in the challenging real world setting.

Full Bayesian modeling approach is particularly useful in comparative effectiveness research. It offers a coherent and flexible framework for incorporating prior knowledge and synthesizing information from different sources.  As a full Bayesian causal inference model, the GPMatch offers a very flexible and general approach to address more complex data types and structures natural to many causal inference problem settings. It can be directly extended to consider multilevel or cluster data structure, and to accommodate complex type of treatment such as multiple level treatment, continuous or composite type of treatment. The model could be extend to time-varying treatment setting without much difficulty by following the g-computation formula framework.  The post-treatment confounding can be addressed by incorporate the confounding variables into the modeling of mean function.  We are already implementing these extensions in an ongoing case study. Although we focused on presenting GPMatch for estimating the average treatment effect (ATE) in this study, the approach is directly applicable to estimation of averaged treatment effect in treated (ATT) and averaged treatment effect in control (ATC). In addition, it can be readily used for modeling treatment effect as a function of pre-specified treatment modifying factors. \cite{sivaganesan2017subgroup} suggested a Bayesian decision theory based approach  for identifying subgroup treatment effect in a randomized trial setting.  With GPMatch, the same idea could be applied to identify subgroup treatment analyzing real world data. Studies are ongoing to evaluate its performances for estimating heterogeneous treatment effect.  The GP regression has been extended to general types of outcomes including binary and count data (\cite{rasmussen2004gaussian}).  Future studies may further investigate its performance under the general types of outcome and data structures. Our simulation focused on comparing with the commonly used causal inference method. Future studies may consider comparisons of our method with other advanced Bayesian methods such as those proposed by \cite{roy2017bayesian} and \cite{Saarela2016}, as well as other advanced non-Bayesian approaches like Targeted MLE (\cite{van2006targeted}). 
  
The GP regression is a very flexible modeling technique, but it is computationally expensive. The time cost associated with GP regression increases at \(n^3\) rate, thus it can be challenging with large sample sizes. The Bayesian Gibbs Sampling algorithm we have used makes it even more demanding in computational resources. Some literature has offered solutions by applying GP to large data, such as \cite{banerjee2008}. Alternatively, one may consider using Bayesian Kernel regression as an approximation. Further studies are needed to improve the computational efficiency and to consider variable selection. It is well known the length scale parameter is hard to estimate. Researchers derived different kinds of priors for GP, for example the objective prior in \cite{Berger2001}, \cite{Kazianka2012}, and \cite{ren2013objective}. \cite{Gelfand2005} suggested using uniform prior for the inverse of the scale parameter in a spatial analysis, but we found that using a prior with preference to smooth surface was more suitable for our purpose. Researchers could also blend their knowledge in the prior to obtain a more efficient estimate. Here we considered squared exponential covariance function but different covariance function such as Matérn could also be considered. Simple block compound symmetry with one correlation coefficient parameter could be used as an alternative covariance matrix. Such blocked covariance set up could be useful particularly for a large sample size and where the data has a reasonable clustering structure, such as in the case of a multi-site study. Future study will explore along this direction. 
\section*{acknowledgements}
This work was supported by an award from the Patient Centered Outcome Research Institute (PCORI ME-1408-19894; PI. B. Huang) and a P\&M pilot award from the Centre for Clinical and Translational Science and Training, which is supported by the National Centre for Advancing Translational Sciences of the National Institutes of Health, under Award Number 5UL1TR001425-03). The Authors declare that there is no conflict of interest. 

\section*{conflict of interest}
The authors declare no conflict of interest .

\bibliography{bibliography}

\end{document}


\appendix

\section*{Supplemental Tables and Figures to "GPMatch: A Bayesian Doubly Robust Approach to Causal Inference with Gaussian Process Covariance Function As a Matching Tool"}

\setcounter{table}{0}
\renewcommand{\thetable}{S\arabic{table}}
\rowcolors{1}{}{light-gray}
\begin{longtable}{rccccccc}
\caption{Results of ATE Estimates from the Single Covariate Simulation Study Setting 1: \(\{\gamma_0, \gamma_1 , \gamma_2, \gamma_3 \}=\{0.5, 0, 0, \sqrt{0.75}\}\).} \label{table:s1}\\
\hline
\showrowcolors
\rowcolor{light-gray2}\textbf{Method} & \textbf{Sample Size} & \textbf{RMSE} & \textbf{MAE} & \textbf{Bias} & \textbf{Rc} & \textbf{\(\boldsymbol{SE_{avg}}\)} & \textbf{\(\boldsymbol{SE_{emp}}\)}\\
\hline
\endfirsthead
\multicolumn{8}{c}%
{\tablename\ \thetable\ -- \textit{Continued from previous page}} \\
\hline
\rowcolor{light-gray2}\textbf{Method} & \textbf{Sample Size} & \textbf{RMSE} & \textbf{MAE} & \textbf{Bias} & \textbf{Rc} & \textbf{\(\boldsymbol{SE_{avg}}\)} & \textbf{\(\boldsymbol{SE_{emp}}\)}\\
\hline
\endhead
\hline \multicolumn{8}{r}{\textit{Continued on next page}} \\
\endfoot
\endlastfoot
GPMatch & 100 & 0.26 & 0.16 & -0.038 & 0.93 & 0.241 & 0.258\\
 & 200 & 0.161 & 0.116 & 0.033 & 0.97 & 0.166 & 0.159\\
 & 400 & 0.122 & 0.085 & -0.005 & 0.96 & 0.118 & 0.123\\
QNT\_PS & 100 & 0.376 & 0.244 & 0.052 & 0.95 & 0.392 & 0.216\\
 & 200 & 0.309 & 0.220 & 0.127 & 0.94 & 0.275 & 0.283\\
 & 400 & 0.238 & 0.159 & 0.096 & 0.92 & 0.201 & 0.219\\
LM & 100 & 0.409 & 0.216 & -0.179 & 0.93 & 0.347 & 0.37\\
 & 200 & 0.291 & 0.183 & -0.119 & 0.89 & 0.25 & 0.266\\
 & 400 & 0.28 & 0.169 & -0.171 & 0.84 & 0.185 & 0.223\\
AIPTW1 & 100 & 0.82 & 0.341 & -0.176 & 0.96 & 0.554 & 0.805\\
 & 200 & 0.765 & 0.294 & -0.209 & 0.98 & 0.504 & 0.74\\
 & 400 & 0.753 & 0.251 & -0.231 & 0.96 & 0.426 & 0.721\\
AIPTW2 & 100 & 0.411 & 0.236 & -0.045 & 0.91 & 0.349 & 0.41\\
 & 200 & 0.288 & 0.203 & 0.029 & 0.93 & 0.268 & 0.288\\
 & 400 & 0.225 & 0.146 & 0.002 & 0.93 & 0.197 & 0.226\\
LM\_PS1 & 100 & 0.367 & 0.239 & -0.109 & 0.91 & 0.332 & 0.352\\
 & 200 & 0.272 & 0.161 & -0.051 & 0.91 & 0.246 & 0.268\\
 & 400 & 0.198 & 0.13 & -0.064 & 0.95 & 0.181 & 0.189\\
LM\_PS2 & 100 & 0.366 & 0.201 & -0.054 & 0.93 & 0.349 & 0.364\\
 & 200 & 0.256 & 0.181 & 0.031 & 0.99 & 0.253 & 0.255\\
 & 400 & 0.185 & 0.136 & -0.004 & 0.95 & 0.186 & 0.186\\
LM\_sp(PS1)  & 100 & 0.264 & 0.186 & -0.054 & 0.91 & 0.241 & 0.26\\
 & 200 & 0.156 & 0.102 & 0.023 & 0.97 & 0.167 & 0.155\\
 & 400 & 0.127 & 0.086 & -0.008 & 0.94 & 0.118 & 0.128\\
LM\_sp(PS2)  & 100 & 0.267 & 0.175 & -0.057 & 0.90 & 0.24 & 0.262\\
 & 200 & 0.155 & 0.11 & 0.02 & 0.98 & 0.167 & 0.154\\
 & 400 & 0.126 & 0.089 & -0.01 & 0.94 & 0.118 & 0.126\\
BART & 100 & 0.27 & 0.156 & -0.026 & 0.95 & 0.257 & 0.27\\
 & 200 & 0.185 & 0.145 & 0.048 & 0.95 & 0.178 & 0.18\\
  & 400 & 0.133 & 0.084 & 0.016 & 0.97 & 0.125 & 0.133\\
\hline 
\end{longtable}
\begin{tablenotes} \footnotesize
\item RMSE = root mean square error; MAE = median absolute error; Bias = Estimate-True; Rc = Rate of coverage by the 95\% interval estimate; \(SE_{avg}\) = average of standard error estimate from all replicate; \(SE_{emp}\) = standard error of ATE estimates from all replicate;
\item GPMatch: Bayesian structural model with Gaussian process prior, only treatment effect is included in the mean function; covariance function includes \(X\). 
\item QNT\_PS: Propensity score sub-classification by quintiles. 
\item AIPTW1 \& AIPTW2: augmented inversed probability of treatment weighting. 
\item LM\_PS1 \& LM\_PS2: linear regression modeling with propensity score adjustment; 
\item LM\_sp(PS1) \& LM\_sp(PS2): linear regression modeling with spline fit propensity score adjustment; 
\item BART: Bayesian additive regression tree. 
\item Propensity scores are estimated using different logistic models, with AIPTW1, LM\_PS1 \& LM\_sp(PS1) use PS estimated using logistic model \(logit A \sim X\) ; and AIPTW2, LM\_PS2 \& LM\_sp(PS2) use PS estimated using logistic model \(logit A \sim X^{1/3}\).  QNT\_PS using either PS estimates produces identical results.  
\end{tablenotes}

\rowcolors{1}{}{light-gray}
\begin{longtable}{rccccccc}
\caption{Results of ATE Estimates from the Single Covariate Simulation Study Setting 2: \(\{\gamma_0, \gamma_1 , \gamma_2, \gamma_3 \}=\{1, 0.15, 0, 0\}\).} \label{table:s2}\\
\hline
\showrowcolors
\rowcolor{light-gray2}\textbf{Method} & \textbf{Sample Size} & \textbf{RMSE} & \textbf{MAE} & \textbf{Bias} & \textbf{Rc} & \textbf{\(\boldsymbol{SE_{avg}}\)} & \textbf{\(\boldsymbol{SE_{emp}}\)}\\
\hline
\endfirsthead
\multicolumn{8}{c}%
{\tablename\ \thetable\ -- \textit{Continued from previous page}} \\
\hline
\rowcolor{light-gray2}\textbf{Method} & \textbf{Sample Size} & \textbf{RMSE} & \textbf{MAE} & \textbf{Bias} & \textbf{Rc} & \textbf{\(\boldsymbol{SE_{avg}}\)} & \textbf{\(\boldsymbol{SE_{emp}}\)}\\
\hline
\endhead
\hline \multicolumn{8}{r}{\textit{Continued on next page}} \\
\endfoot
\endlastfoot
GPMatch & 100 & 0.237 & 0.152 & 0.013 & 0.97 & 0.243 & 0.238\\
 & 200 & 0.175 & 0.114 & 0.007 & 0.94 & 0.169 & 0.175\\
 & 400 & 0.117 & 0.084 & 0.001 & 0.96 & 0.117 & 0.118\\
QNT\_PS & 100 & 0.436 & 0.271 & 0.089 & 0.95 & 0.466 & 0.429\\
 & 200 & 0.301 & 0.210 & 0.103 & 0.98 & 0.287 & 0.284\\
 & 400 & 0.254 & 0.171 & 0.096 & 0.88 & 0.209 & 0.236\\
LM & 100 & 0.427 & 0.255 & -0.214 & 0.93 & 0.399 & 0.371\\
 & 200 & 0.348 & 0.174 & -0.164 & 0.93 & 0.26 & 0.309\\
 & 400 & 0.318 & 0.166 & -0.198 & 0.81 & 0.191 & 0.25\\
AIPTW1 & 100 & 0.933 & 0.378 & -0.226 & 0.96 & 0.671 & 0.91\\
 & 200 & 3.853 & 0.246 & -0.478 & 0.95 & 0.861 & 3.842\\
 & 400 & 1.25 & 0.213 & -0.396 & 0.98 & 0.565 & 1.192\\
AIPTW2 & 100 & 0.413 & 0.306 & -0.029 & 0.96 & 0.411 & 0.414\\
 & 200 & 0.345 & 0.156 & -0.001 & 0.96 & 0.281 & 0.346\\
 & 400 & 0.244 & 0.124 & -0.021 & 0.97 & 0.221 & 0.244\\
LM\_PS1  & 100 & 0.352 & 0.265 & -0.087 & 0.97 & 0.368 & 0.343\\
 & 200 & 0.273 & 0.181 & -0.055 & 0.93 & 0.251 & 0.269\\
 & 400 & 0.192 & 0.11 & -0.082 & 0.92 & 0.188 & 0.174\\
LM\_PS2 & 100 & 0.337 & 0.251 & -0.018 & 0.98 & 0.397 & 0.339\\
 & 200 & 0.252 & 0.154 & -0.004 & 0.97 & 0.262 & 0.253\\
 & 400 & 0.175 & 0.101 & -0.004 & 0.98 & 0.192 & 0.176\\
LM\_sp(PS1)  & 100 & 0.237 & 0.158 & -0.006 & 0.97 & 0.242 & 0.238\\
 & 200 & 0.171 & 0.109 & -0.004 & 0.94 & 0.169 & 0.172\\
 & 400 & 0.118 & 0.083 & -0.003 & 0.96 & 0.118 & 0.118\\
LM\_sp(PS2)  & 100 & 0.248 & 0.163 & 0.002 & 0.96 & 0.243 & 0.249\\
 & 200 & 0.171 & 0.103 & -0.001 & 0.95 & 0.169 & 0.172\\
 & 400 & 0.116 & 0.087 & -0.006 & 0.96 & 0.118 & 0.117\\
Bart & 100 & 0.286 & 0.176 & 0.054 & 0.95 & 0.266 & 0.283\\
 & 200 & 0.182 & 0.115 & 0.034 & 0.96 & 0.18 & 0.18\\
  & 400 & 0.161 & 0.085 & 0.01 & 0.93 & 0.127 & 0.161\\
\hline 
\end{longtable}
\begin{tablenotes} \footnotesize
\item RMSE = root mean square error; MAE = median absolute error; Bias = Estimate-True; Rc = Rate of coverage by the 95\% interval estimate; \(SE_{avg}\) = average of standard error estimate from all replicate; \(SE_{emp}\) = standard error of ATE estimates from all replicate;
\item GPMatch: Bayesian structural model with Gaussian process prior, only treatment effect is included in the mean function; covariance function includes \(X\). 
\item QNT\_PS: Propensity score sub-classification by quintiles. 
\item AIPTW1 \& AIPTW2: augmented inversed probability of treatment weighting. 
\item LM\_PS1 \& LM\_PS2: linear regression modeling with propensity score adjustment; 
\item LM\_sp(PS1) \& LM\_sp(PS2): linear regression modeling with spline fit propensity score adjustment; 
\item BART: Bayesian additive regression tree. 
\item Propensity scores are estimated using different logistic models, with AIPTW1, LM\_PS1 \& LM\_sp(PS1) use PS estimated using logistic model \(logit A \sim X\) ; and AIPTW2, LM\_PS2 \& LM\_sp(PS2) use PS estimated using logistic model \(logit A \sim X^{1/3}\).  QNT\_PS using either PS estimates produces identical results.  
\end{tablenotes}
\bigskip

\rowcolors{1}{}{light-gray}
\begin{longtable}{rccccccc}
\caption{Results of ATE Estimates from the Single Covariate Simulation Study Setting 3: \(\{\gamma_0, \gamma_1 , \gamma_2, \gamma_3 \}=\{0.5, 0, 0.7, \sqrt{0.75}\}\).} \label{table:s2}\\
\hline
\showrowcolors
\rowcolor{light-gray2}\textbf{Method} & \textbf{Sample Size} & \textbf{RMSE} & \textbf{MAE} & \textbf{Bias} & \textbf{Rc} & \textbf{\(\boldsymbol{SE_{avg}}\)} & \textbf{\(\boldsymbol{SE_{emp}}\)}\\
\hline
\endfirsthead
\multicolumn{8}{c}%
{\tablename\ \thetable\ -- \textit{Continued from previous page}} \\
\hline
\rowcolor{light-gray2}\textbf{Method} & \textbf{Sample Size} & \textbf{RMSE} & \textbf{MAE} & \textbf{Bias} & \textbf{Rc} & \textbf{\(\boldsymbol{SE_{avg}}\)} & \textbf{\(\boldsymbol{SE_{emp}}\)}\\
\hline
\endhead
\hline \multicolumn{8}{r}{\textit{Continued on next page}} \\
\endfoot
\endlastfoot
GPMatch & 100 & 0.231 & 0.156 & 0.009 & 0.96 & 0.238 & 0.232\\
 & 200 & 0.166 & 0.107 & -0.003 & 0.93 & 0.164 & 0.167\\
 & 400 & 0.115 & 0.088 & 0.003 & 0.96 & 0.114 & 0.115\\
QNT\_PS X & 100 & 0.554 & 0.443 & 0.408 & 0.91 & 0.442 & 0.377\\
 & 200 & 0.364 & 0.282 & 0.263 & 0.92 & 0.298 & 0.252\\
 & 400 & 0.328 & 0.278 & 0.260 & 0.77 & 0.211 & 0.201\\
LM & 100 & 0.364 & 0.256 & 0.121 & 0.95 & 0.350 & 0.345\\
 & 200 & 0.279 & 0.179 & 0.155 & 0.88 & 0.256 & 0.233\\
 & 400 & 0.230 & 0.192 & 0.164 & 0.87 & 0.181 & 0.162\\
AIPTW1 & 100 & 0.533 & 0.432 & 0.413 & 0.81 & 0.355 & 0.339\\
 & 200 & 0.544 & 0.512 & 0.462 & 0.58 & 0.282 & 0.289\\
 & 400 & 0.504 & 0.466 & 0.469 & 0.24 & 0.196 & 0.185\\
AIPTW2 & 100 & 0.507 & 0.430 & 0.401 & 0.79 & 0.322 & 0.311\\
 & 200 & 0.481 & 0.444 & 0.421 & 0.54 & 0.234 & 0.234\\
 & 400 & 0.456 & 0.424 & 0.428 & 0.23 & 0.162 & 0.157\\
LM\_PS1 & 100 & 0.374 & 0.234 & 0.213 & 0.80 & 0.261 & 0.309\\
 & 200 & 0.382 & 0.281 & 0.289 & 0.67 & 0.191 & 0.251\\
 & 400 & 0.335 & 0.282 & 0.288 & 0.44 & 0.133 & 0.171\\
LM\_PS2 & 100 & 0.500 & 0.384 & 0.386 & 0.84 & 0.339 & 0.319\\
 & 200 & 0.495 & 0.432 & 0.427 & 0.61 & 0.251 & 0.251\\
 & 400 & 0.464 & 0.416 & 0.433 & 0.30 & 0.177 & 0.166\\
LM\_sp(PS1) & 100 & 0.235 & 0.162 & -0.001 & 0.94 & 0.235 & 0.237\\
 & 200 & 0.170 & 0.114 & -0.012 & 0.92 & 0.165 & 0.170\\
 & 400 & 0.115 & 0.090 & 0.002 & 0.96 & 0.115 & 0.116\\
LM\_sp(PS2) & 100 & 0.232 & 0.161 & 0.005 & 0.97 & 0.238 & 0.233\\
 & 200 & 0.167 & 0.114 & -0.008 & 0.95 & 0.167 & 0.167\\
 & 400 & 0.115 & 0.094 & 0.002 & 0.97 & 0.117 & 0.115\\
Bart & 100 & 0.274 & 0.191 & 0.114 & 0.97 & 0.265 & 0.251\\
 & 200 & 0.195 & 0.116 & 0.071 & 0.93 & 0.188 & 0.182\\
 & 400 & 0.138 & 0.095 & 0.057 & 0.94 & 0.136 & 0.126\\
\hline 
\end{longtable}
\begin{tablenotes} \footnotesize
\item RMSE = root mean square error; MAE = median absolute error; Bias = Estimate-True; Rc = Rate of coverage by the 95\% interval estimate; \(SE_{avg}\) = average of standard error estimate from all replicate; \(SE_{emp}\) = standard error of ATE estimates from all replicate;
\item GPMatch: Bayesian structural model with Gaussian process prior, only treatment effect is included in the mean function; covariance function includes \(X\). 
\item QNT\_PS: Propensity score sub-classification by quintiles. 
\item AIPTW1 \& AIPTW2: augmented inversed probability of treatment weighting. 
\item LM\_PS1 \& LM\_PS2: linear regression modeling with propensity score adjustment; 
\item LM\_sp(PS1) \& LM\_sp(PS2): linear regression modeling with spline fit propensity score adjustment; 
\item BART: Bayesian additive regression tree. 
\item Propensity scores are estimated using different logistic models, with AIPTW1, LM\_PS1 \& LM\_sp(PS1) use PS estimated using logistic model \(logit A \sim X\) ; and AIPTW2, LM\_PS2 \& LM\_sp(PS2) use PS estimated using logistic model \(logit A \sim X^{1/3}\).  QNT\_PS using either PS estimates produces identical results.  
\end{tablenotes}
\bigskip

\rowcolors{1}{}{light-gray}
\begin{longtable}{rccccccc}
\caption{Results of ATE Estimates from the Single Covariate Simulation Study Setting 4: \(\{\gamma_0, \gamma_1 , \gamma_2, \gamma_3 \}=\{1, 0.15, 0.7, 0\}\).} \label{table:s2}\\
\hline
\showrowcolors
\rowcolor{light-gray2}\textbf{Method} & \textbf{Sample Size} & \textbf{RMSE} & \textbf{MAE} & \textbf{Bias} & \textbf{Rc} & \textbf{\(\boldsymbol{SE_{avg}}\)} & \textbf{\(\boldsymbol{SE_{emp}}\)}\\
\hline
\endfirsthead
\multicolumn{8}{c}%
{\tablename\ \thetable\ -- \textit{Continued from previous page}} \\
\hline
\rowcolor{light-gray2}\textbf{Method} & \textbf{Sample Size} & \textbf{RMSE} & \textbf{MAE} & \textbf{Bias} & \textbf{Rc} & \textbf{\(\boldsymbol{SE_{avg}}\)} & \textbf{\(\boldsymbol{SE_{emp}}\)}\\
\hline
\endhead
\hline \multicolumn{8}{r}{\textit{Continued on next page}} \\
\endfoot
\endlastfoot
GPMatch & 100 & 0.226 & 0.140 & 0.022 & 0.97 & 0.238 & 0.226\\
 & 200 & 0.164 & 0.105 & 0.024 & 0.98 & 0.169 & 0.163\\
 & 400 & 0.104 & 0.073 & 0.009 & 0.96 & 0.114 & 0.104\\
QNT\_PS & 100 & 0.579 & 0.477 & 0.437 & 0.90 & 0.452 & 0.382\\
 & 200 & 0.369 & 0.312 & 0.279 & 0.87 & 0.308 & 0.244\\
 & 400 & 0.300 & 0.268 & 0.245 & 0.78 & 0.199 & 0.174\\
LM & 100 & 0.323 & 0.236 & 0.175 & 0.96 & 0.362 & 0.273\\
 & 200 & 0.292 & 0.214 & 0.187 & 0.90 & 0.259 & 0.226\\
 & 400 & 0.213 & 0.170 & 0.165 & 0.87 & 0.174 & 0.136\\
AIPTW1 & 100 & 0.596 & 0.456 & 0.461 & 0.78 & 0.377 & 0.380\\
 & 200 & 0.523 & 0.457 & 0.473 & 0.48 & 0.262 & 0.223\\
 & 400 & 0.474 & 0.459 & 0.433 & 0.32 & 0.210 & 0.192\\
AIPTW2 & 100 & 0.541 & 0.452 & 0.438 & 0.76 & 0.335 & 0.319\\
 & 200 & 0.488 & 0.451 & 0.446 & 0.48 & 0.231 & 0.199\\
 & 400 & 0.418 & 0.407 & 0.393 & 0.31 & 0.166 & 0.143\\
LM\_PS1 & 100 & 0.398 & 0.264 & 0.256 & 0.84 & 0.266 & 0.307\\
 & 200 & 0.376 & 0.297 & 0.288 & 0.64 & 0.191 & 0.244\\
 & 400 & 0.340 & 0.305 & 0.303 & 0.35 & 0.132 & 0.155\\
LM\_PS2 & 100 & 0.563 & 0.421 & 0.444 & 0.81 & 0.352 & 0.349\\
 & 200 & 0.498 & 0.471 & 0.453 & 0.50 & 0.253 & 0.207\\
 & 400 & 0.411 & 0.400 & 0.388 & 0.32 & 0.175 & 0.135\\
LM\_sp(PS1) & 100 & 0.226 & 0.140 & 0.004 & 0.970 & 0.237 & 0.227\\
 & 200 & 0.138 & 0.097 & 0.007 & 0.98 & 0.165 & 0.138\\
 & 400 & 0.104 & 0.067 & 0.006 & 0.97 & 0.116 & 0.104\\
LM\_sp(PS2) & 100 & 0.235 & 0.147 & 0.009 & 0.95 & 0.240 & 0.236\\
 & 200 & 0.144 & 0.112 & 0.012 & 0.96 & 0.167 & 0.144\\
 & 400 & 0.401 & 0.361 & 0.361 & 0.39 & 0.165 & 0.174\\
Bart & 100 & 0.262 & 0.183 & 0.105 & 0.97 & 0.271 & 0.241\\
 & 200 & 0.168 & 0.117 & 0.076 & 0.97 & 0.190 & 0.150\\
 & 400 & 0.130 & 0.097 & 0.064 & 0.96 & 0.136 & 0.114\\
\hline 
\end{longtable}
\begin{tablenotes} \footnotesize
\item RMSE = root mean square error; MAE = median absolute error; Bias = Estimate-True; Rc = Rate of coverage by the 95\% interval estimate; \(SE_{avg}\) = average of standard error estimate from all replicate; \(SE_{emp}\) = standard error of ATE estimates from all replicate;
\item GPMatch: Bayesian structural model with Gaussian process prior, only treatment effect is included in the mean function; covariance function includes \(X\). 
\item QNT\_PS: Propensity score sub-classification by quintiles. 
\item AIPTW1 \& AIPTW2: augmented inversed probability of treatment weighting. 
\item LM\_PS1 \& LM\_PS2: linear regression modeling with propensity score adjustment; 
\item LM\_sp(PS1) \& LM\_sp(PS2): linear regression modeling with spline fit propensity score adjustment; 
\item BART: Bayesian additive regression tree. 
\item Propensity scores are estimated using different logistic models, with AIPTW1, LM\_PS1 \& LM\_sp(PS1) use PS estimated using logistic model \(logit A \sim X\) ; and AIPTW2, LM\_PS2 \& LM\_sp(PS2) use PS estimated using logistic model \(logit A \sim X^{1/3}\).  QNT\_PS using either PS estimates produces identical results.  
\end{tablenotes}
\bigskip

\setcounter{figure}{0}
\renewcommand{\thefigure}{S\arabic{figure}} 
\begin{figure}[!htb]
\centering
\caption{Comparisons of root mean square error (RMSE), and median absolute error (MAE) of the ATE Estimates by Different  Methods Across Different Sample Sizes under the Simulation Setting 1: \(\{\gamma_0, \gamma_1 , \gamma_2, \gamma_3 \}=\{0.5, 0, 0, \sqrt{0.75}\}\)}
\includegraphics[width=12cm]{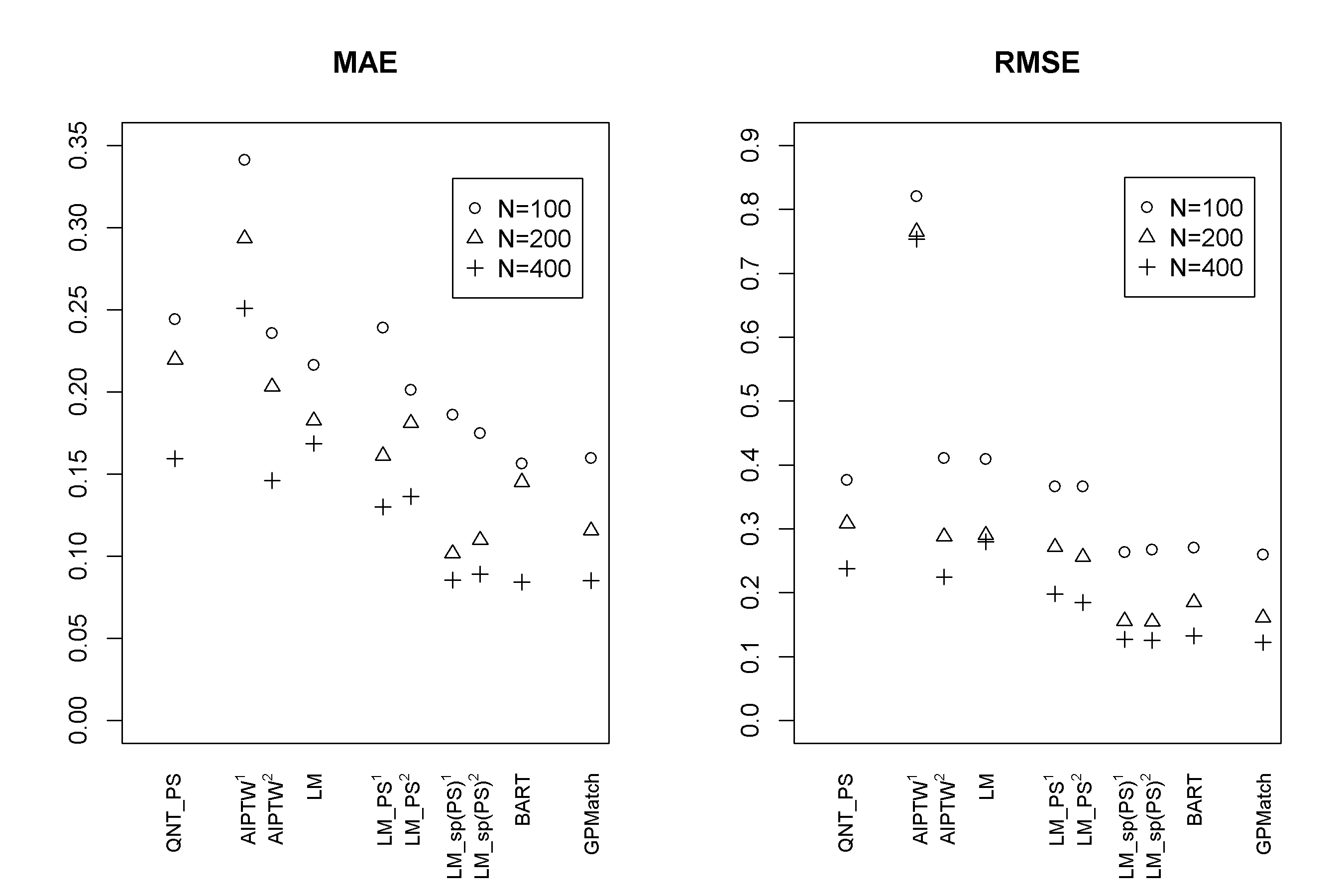}
\caption*{\(^1\) Propensity score estimated using logistic regression on \(logit A \sim X\).
\(^2\) Propensity score estimated using logistic regression on \(logit A \sim X^{1/3}\).
GPMatch: Bayesian structural model with Gaussian process prior.
QNT\_PS: Propensity score sub-classification by quintiles.
AIPTW: augmented inversed probability of treatment weighting;  
LM: linear regression modeling \(Y \sim X\); 
LM\_PS: linear regression modeling with propensity score adjustment.  
LM\_sp(PS): linear regression modeling with spline fit propensity score adjustment.
BART: Bayesian additive regression tree. 
}
\label{fig:s1}
\end{figure}

\begin{figure}[bt]
\centering
\caption{Comparisons of root mean square error (RMSE), and median absolute error (MAE) of the ATE Estimates by Different  Methods Across Different Sample Sizes under the Simulation Setting 2: \(\{\gamma_0, \gamma_1 , \gamma_2, \gamma_3 \}=\{1, 0.15, 0, 0\}\)}
\includegraphics[width=12cm]{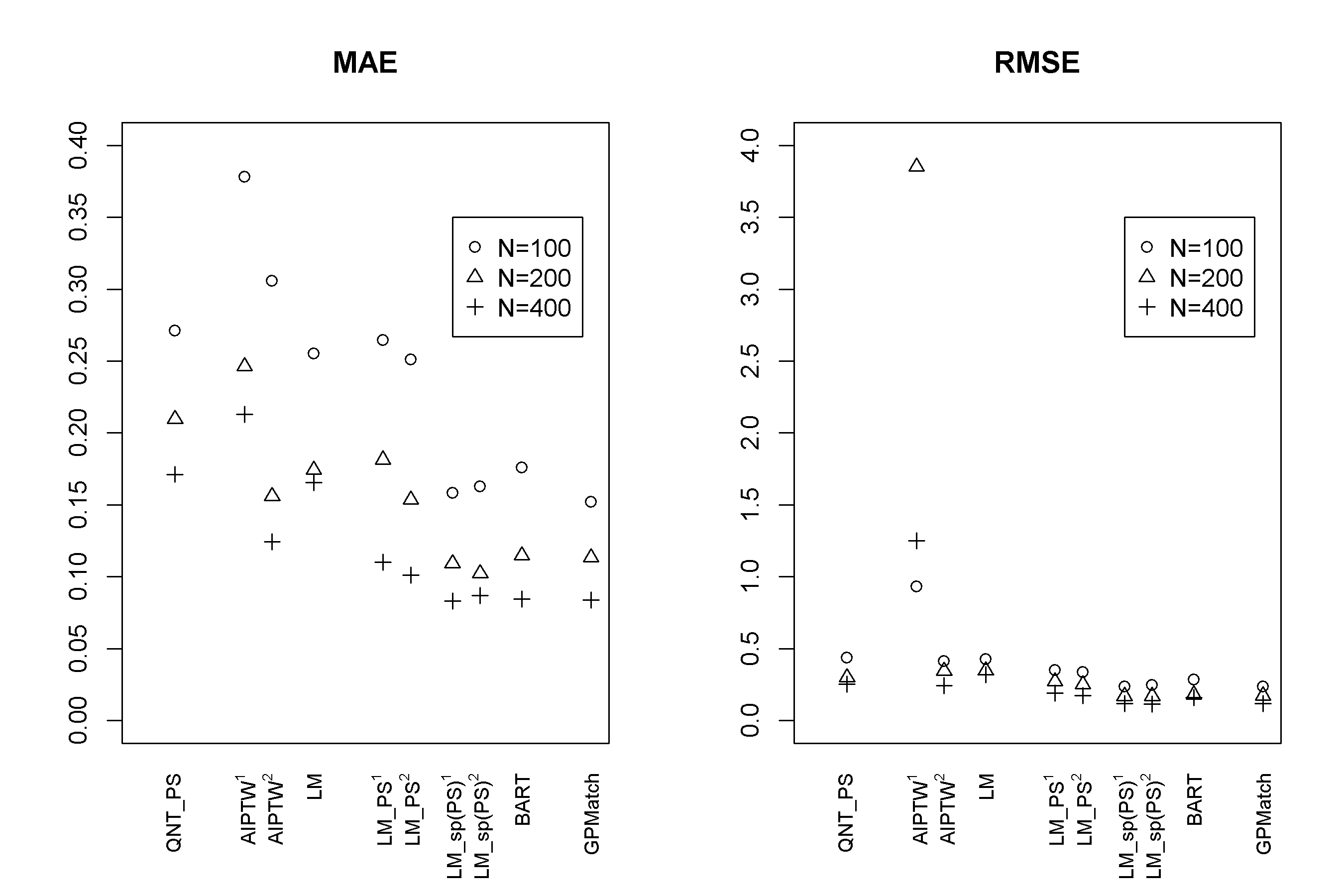}
\caption*{\(^1\) Propensity score estimated using logistic regression on \(logit A \sim X\).
\(^2\) Propensity score estimated using logistic regression on \(logit A \sim X^{1/3}\).
GPMatch: Bayesian structural model with Gaussian process prior.
QNT\_PS: Propensity score sub-classification by quintiles.
AIPTW: augmented inversed probability of treatment weighting;  
LM: linear regression modeling \(Y \sim X\); 
LM\_PS: linear regression modeling with propensity score adjustment.  
LM\_sp(PS): linear regression modeling with spline fit propensity score adjustment.
BART: Bayesian additive regression tree. 
}
\label{fig:s2}
\end{figure}

\begin{figure}[!htb]
\centering
\caption{Comparisons of root mean square error (RMSE), and median absolute error (MAE) of the ATE Estimates by Different  Methods Across Different Sample Sizes under the Simulation Setting 3: \(\{\gamma_0, \gamma_1 , \gamma_2, \gamma_3 \}=\{0.5, 0, 0.7, \sqrt{0.75}\}\)}
\includegraphics[width=12cm]{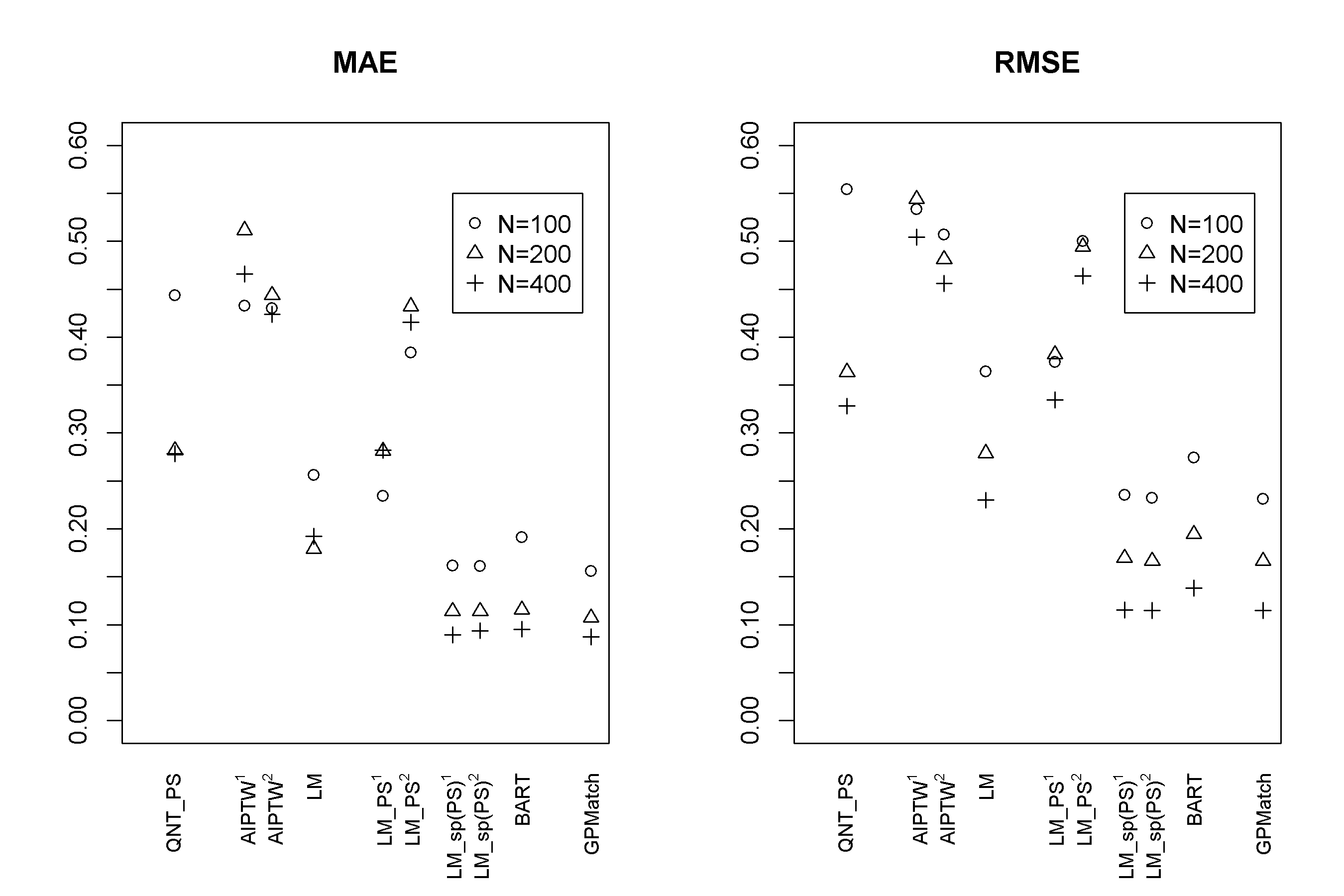}
\caption*{\(^1\) Propensity score estimated using logistic regression on \(logit A \sim X\).
\(^2\) Propensity score estimated using logistic regression on \(logit A \sim X^{1/3}\).
GPMatch: Bayesian structural model with Gaussian process prior.
QNT\_PS: Propensity score sub-classification by quintiles.
AIPTW: augmented inversed probability of treatment weighting;  
LM: linear regression modeling \(Y \sim X\); 
LM\_PS: linear regression modeling with propensity score adjustment.  
LM\_sp(PS): linear regression modeling with spline fit propensity score adjustment.
BART: Bayesian additive regression tree. 
}
\label{fig:s1a}
\end{figure}

\begin{figure}[bt]
\centering
\caption{Comparisons of root mean square error (RMSE), and median absolute error (MAE) of the ATE Estimates by Different  Methods Across Different Sample Sizes under the Simulation Setting 4: \(\{\gamma_0, \gamma_1 , \gamma_2, \gamma_3 \}=\{1, 0.15, 0.7, 0\}\)}
\includegraphics[width=12cm]{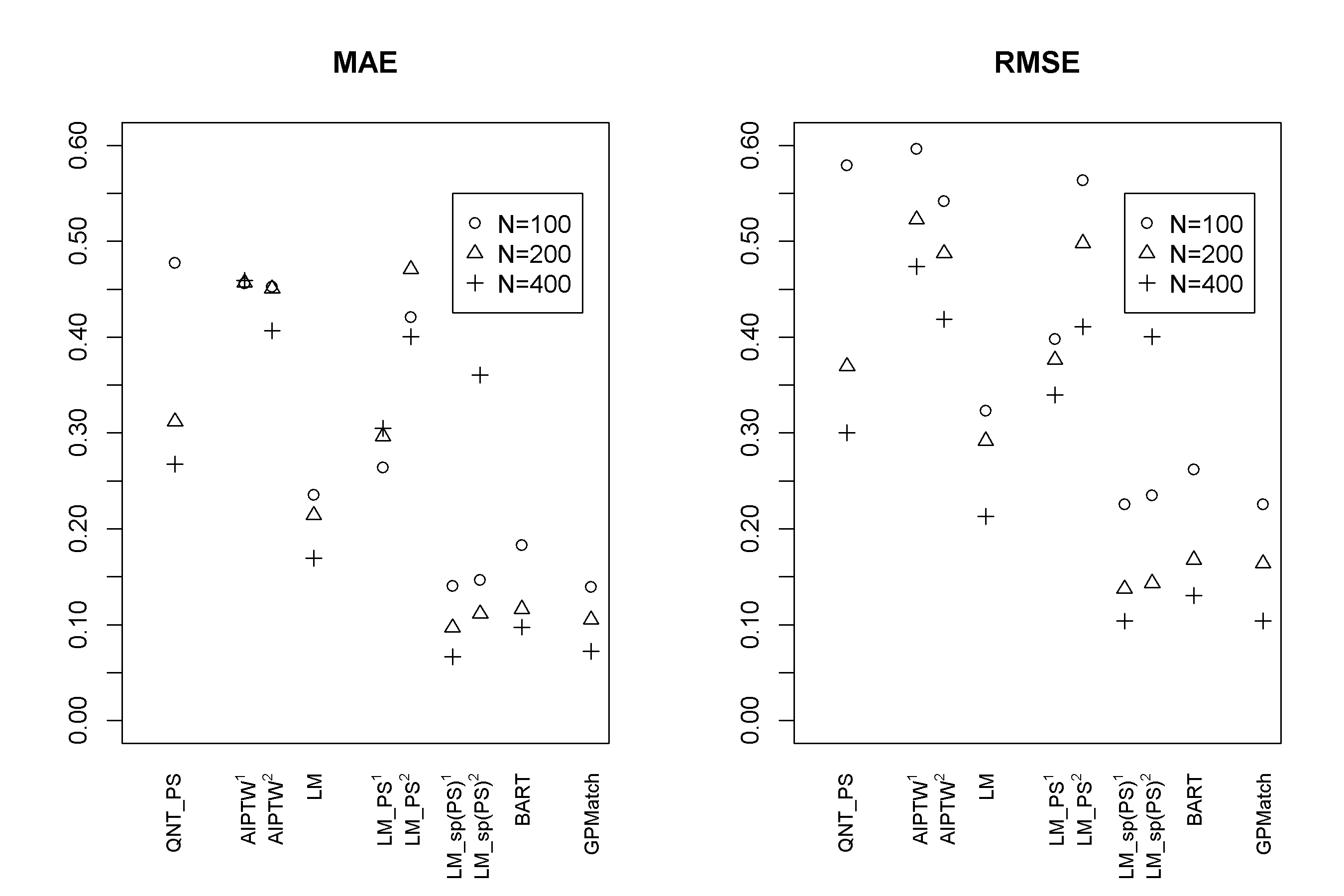}
\caption*{\(^1\) Propensity score estimated using logistic regression on \(logit A \sim X\).
\(^2\) Propensity score estimated using logistic regression on \(logit A \sim X^{1/3}\).
GPMatch: Bayesian structural model with Gaussian process prior.
QNT\_PS: Propensity score sub-classification by quintiles.
AIPTW: augmented inversed probability of treatment weighting;  
LM: linear regression modeling \(Y \sim X\); 
LM\_PS: linear regression modeling with propensity score adjustment.  
LM\_sp(PS): linear regression modeling with spline fit propensity score adjustment.
BART: Bayesian additive regression tree. 
}
\label{fig:s2b}
\end{figure}

\begin{figure}[bt]
\centering
\caption{Distribution of the Estimated by Different Sample Sizes ATE from GPMatch under the Kang and Shafer Dual Misspecifcation Setting. Upper panel presents the results of GPMatch with the treatment effect only in the mean function model; lower panel presents the results of GPMatch with the treatment effect and the \(X_1-X_4\) in the mean function model. Both included \(X_1-X_4\) in the covariate function.}
\includegraphics[width=12cm]{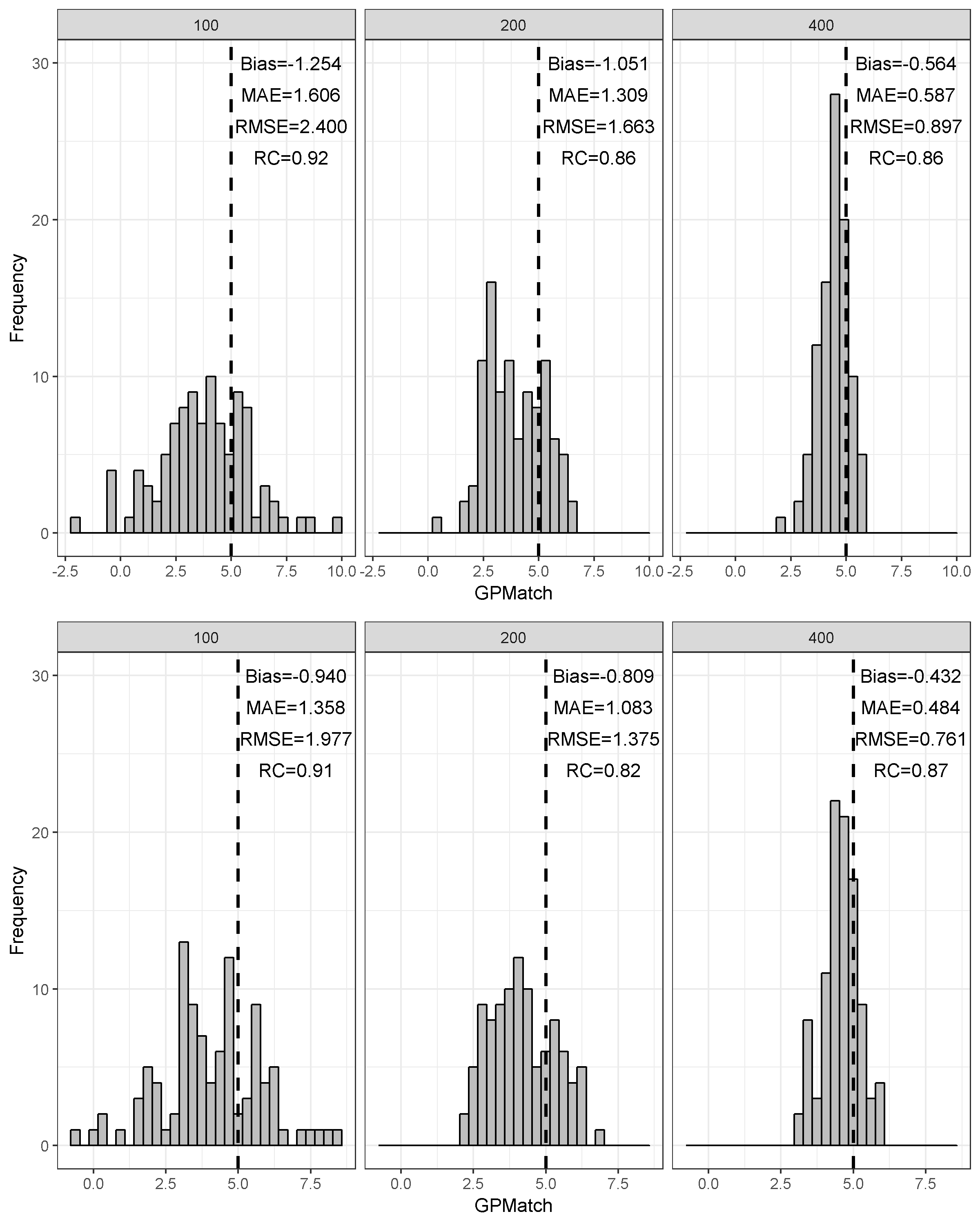}

\label{fig:s3}
\end{figure}

\begin{figure}[bt]
\centering
\caption{Distributions of key covariates in unweighted and weighted samples using inverse probability weighting of propensity scores for the case study}
\includegraphics[width=12cm]{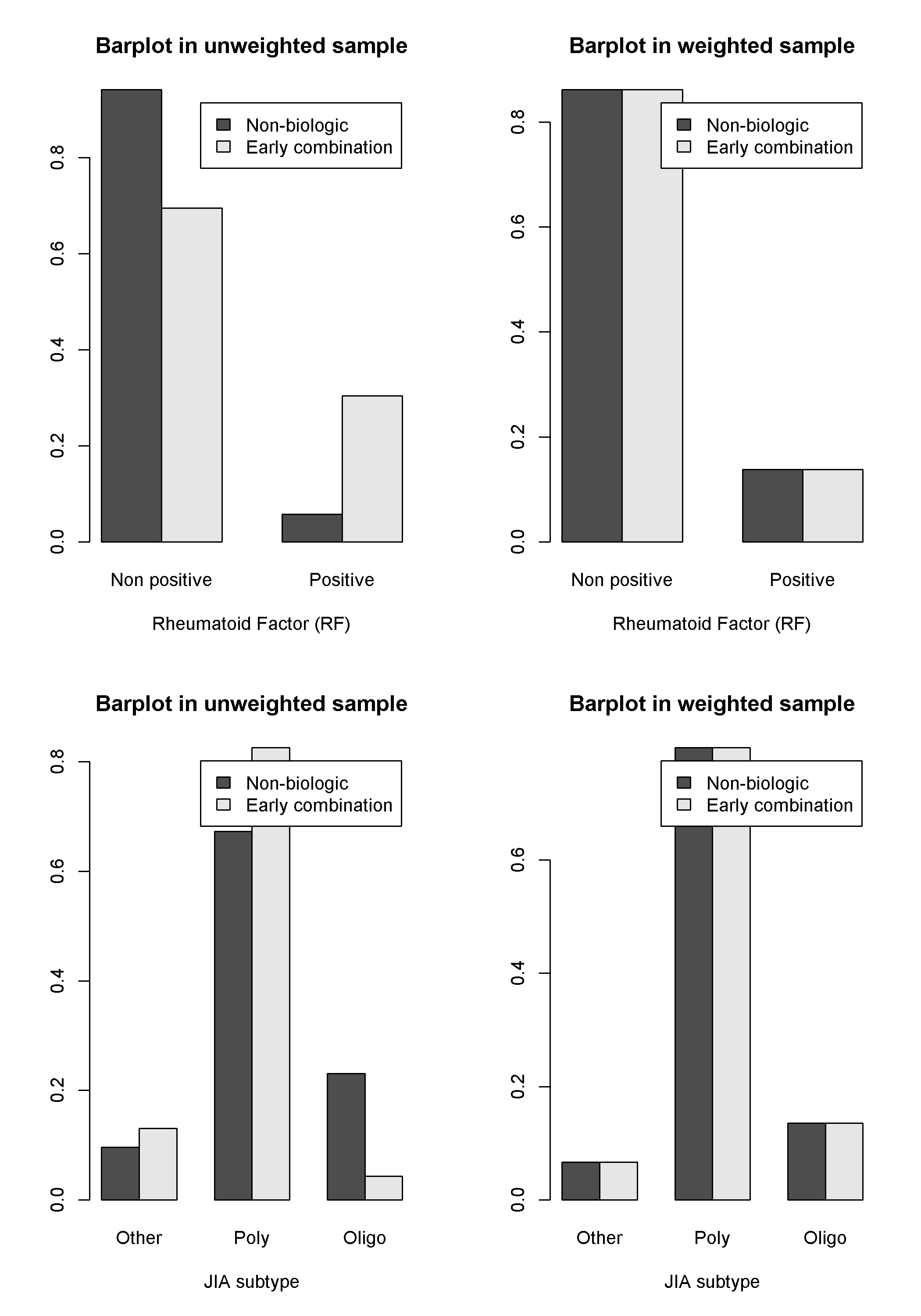}
\label{fig:s4a}
\end{figure}

\begin{figure}[bt]
\centering
\includegraphics[width=12cm]{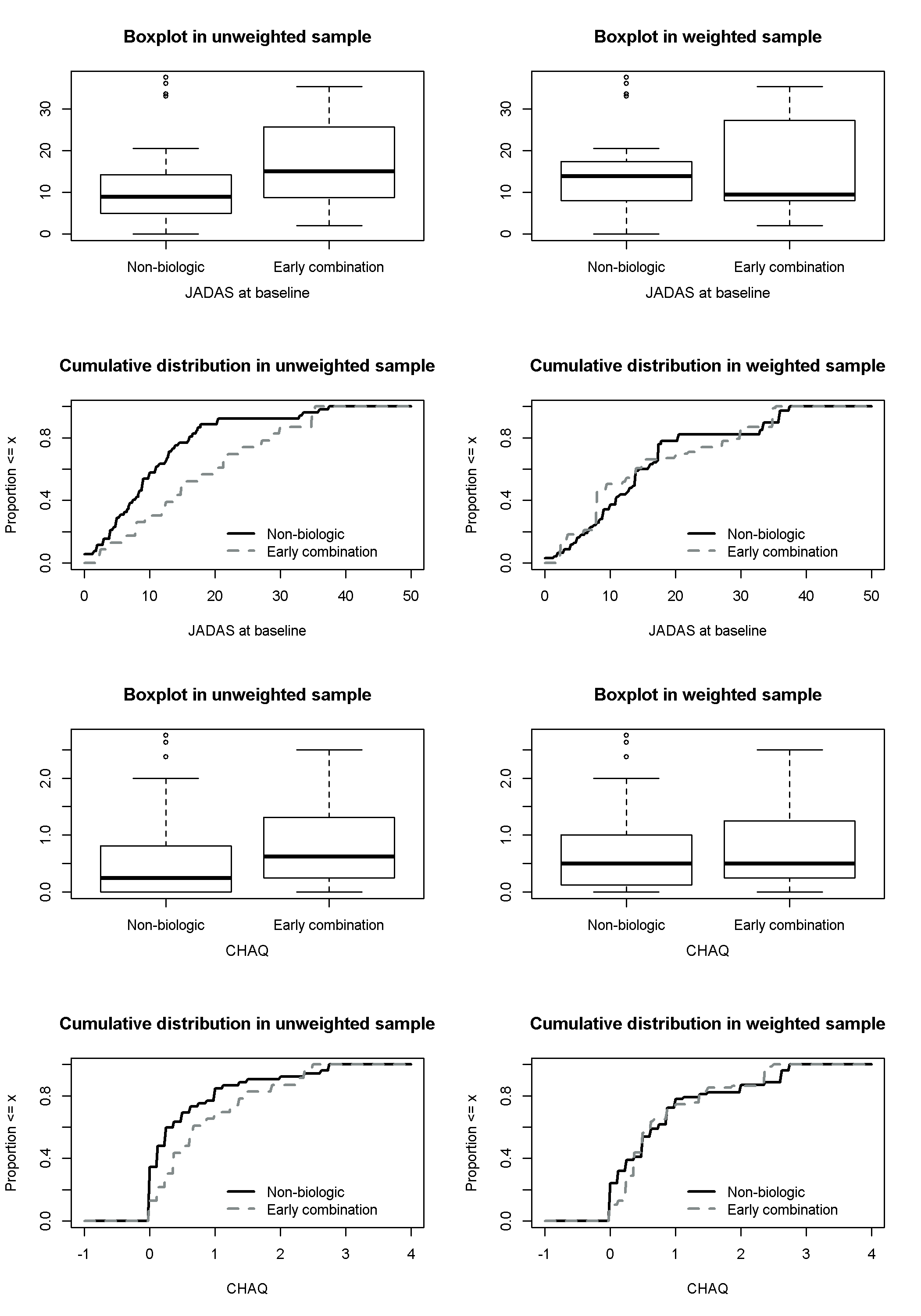}
\label{fig:s4b}
\end{figure}